%% file: neurips_2023_games_lower_bound.tex
\documentclass{article}

\PassOptionsToPackage{numbers, compress}{natbib}

\usepackage[final]{neurips_2023}

\usepackage[utf8]{inputenc} 
\usepackage[T1]{fontenc}    
\usepackage{hyperref}       
\usepackage{url}            
\usepackage{amsfonts}       
\usepackage{nicefrac}       
\usepackage{microtype}      
\usepackage{xcolor}         
\usepackage{comment}

\input{preamble.tex}

\title{Towards Characterizing the First-order Query Complexity of Learning (Approximate) Nash Equilibria in Zero-sum Matrix Games}

\author{%
  Hedi Hadiji 
    \\
  Laboratoire des signaux et systèmes\\
  Univ. Paris-Saclay, CNRS, CentraleSupélec
  \And
  Sarah Sachs, Tim van Even \\
  Korteweg-de Vries Institute for Mathematics \\
  University of Amsterdam
  \And
  Wouter Koolen \\
  Centrum Wiskunde \& Informatica and University of Twente
}

\begin{document}

\maketitle

\begin{abstract}
  In the first-order query model for zero-sum $K\times K$ matrix games, players observe the expected pay-offs for all their possible actions under the randomized action played by their opponent. This classical model has received renewed interest after the discovery by \citeauthor{rakhlin2013optimization} that $\epsilon$-approximate Nash equilibria can be computed efficiently from $O(\ln K / \epsilon)$ instead of $O(\ln K / \epsilon^2)$ queries. Surprisingly, the optimal number of such queries, as a function of both $\epsilon$ and $K$, is not known. We make progress on this question on two fronts. First, we fully characterise the query complexity of learning exact equilibria ($\epsilon=0$), by showing that they require a number of queries that is linear in $K$, which means that it is essentially as hard as querying the whole matrix, which can also be done with $K$ queries. Second, for $\epsilon > 0$, the current query complexity upper bound stands at $O(\min(\ln(K) / \epsilon , K))$. We argue that, unfortunately, obtaining a matching lower bound is not possible with existing techniques: we prove that no lower bound can be derived by constructing hard matrices whose entries take values in a known countable set, because such matrices can be fully identified by a single query. This rules out, for instance, reducing to an optimization problem over the hypercube by encoding it as a binary payoff matrix. We then introduce a new technique for lower bounds, which allows us to obtain lower bounds of order $\tilde\Omega(\log(\frac{1}{K\epsilon}))$ for any $\epsilon \leq 1 / (cK^4)$, where $c$ is a constant independent of $K$. We further discuss possible future directions to improve on our techniques in order to close the gap with the upper bounds.
\end{abstract}

\section{Introduction}

Computing the saddle point
\begin{equation*}
  \min_{x \in \cX} \max_{y \in \cY} f(x, y) \, 
\end{equation*}
for convex-concave functions $f: \cX \times \cY \to \R$ is of general
interest throughout optimization, economics and machine learning. We
study the computation of an approximate saddle point $(x_\star,
y_\star)$, satisfying $\max_{y \in \cY} f(x_\star,y)  - \min_{x \in
\cX} f(x,y_\star) \leq 2\epsilon$ for some given $\eps \geq 0$.
Starting with an unknown $f$ from a known class $\mathcal F$, we
consider sequential learning in the first-order feedback model, where
the learner gets to observe gradients of the objective. Formally, each
query $(x, y)$ results in feedback $(\nabla_x f(x,y), \nabla_y f(x,y))$.
We are interested in the following question:
\begin{center}
 \emph{How many first-order queries are necessary and sufficient for a sequential learner to output an approximate saddle point for any $f \in \mathcal F$?}
\end{center}
Characterizing the query complexity of learning saddle points is of theoretical interest for understanding the hardness of computing equilibria, and for certifying the optimality of upper bounds.

In this work, we restrict our attention to the special case of
zero-sum matrix games, where $\cX$ and $\cY$ are finite-dimensional
probability simplices and $f$ is bilinear. For this canonical setting,
the optimal query complexity is, surprisingly, unresolved. Indeed, computation algorithms are known since
\cite{brown1951iterative-solut}, up to \cite{rakhlin2013optimization}, while lower bounds remain elusive, leaving the optimal query complexity still unknown.

Obtaining lower bounds is not only of fundamental interest in itself,
but may also lead to interesting new techniques, since none of the existing proof strategies are applicable. New ideas provided here could prove useful to tackle other problems.

\subsection{Contributions}
We make progress towards characterizing the first-order query complexity of approximate Nash equilibria in finite-action zero-sum games, as a function of the number of actions $K$ and approximation level $\eps$. Our contributions are the following:
\paragraph{Lower bounds} 
We provide the first lower bounds on the first-order query complexity
for zero-sum matrix games with bounded entries. We show that $K/2-1$
queries are necessary to compute an exact equilibrium
(Theorem~\ref{thm:exact_lower_bound} in Section~\ref{sec:lowerbounds}),
and at least $\Omega(\log( 1 / \eps K) / \log (K))$ queries are required
for an $\eps$-equilibrium if $\eps \leq 1 / K^4$
(Theorem~\ref{thm:approx-lower-bound} in Section~\ref{sec:lowerbounds}).
More than the concrete rates, we believe that the structure of our proof
is of particular interest, and might be of use beyond this specific
setting.
\paragraph{Upper Bounds}
We show that if the game matrix has entries in a known countable set, then the learner can recover the full matrix in one single first-order query (Theorem~\ref{thm:DoomedApproachCoutable} in Section~\ref{sec:upper-bounds}). This result can be interpreted both as an upper bound on the query complexity for matrices on countable entry sets, and as an impossibility result, precluding the use of simple predefined sets of matrices as candidate objectives for proving lower bounds. This, together with the lack of rotational invariance of the action sets, sheds light on why this setting turns out to be surprisingly resistant to proving lower bounds using well-established techniques.

\subsection{Related Work on Lower Bounds: Lower Bounds}\label{sec:related-work}

Below we review the literature on lower bound techniques for related problems and settings. We find that none of the existing techniques apply to our setting, for three particular reasons. They either:
\begin{itemize}[leftmargin=*]
  \item[--] Build matrices from a set of matrices with entries lying in a finite alphabet. In Theorem~\ref{thm:DoomedApproachCoutable}, we show that, for game
  matrices with entries in a known countable set, a single first-order
  query suffices to infer the exact Nash equilibrium. This shows in particular that any such matrix construction cannot lead to a lower bound in the first-order model.
  \item[--] Rely crucially on rotational invariance of the action set, which limits the approach to the $\ell_2$-constrained or unconstrained cases, ruling out the probability simplex.
  \item[--] Assume that the players select actions in the span of the observed gradients. This span assumption is not suited when the action set is not an $\ell_2$ ball, in which case the span of the gradients has no natural embedding into the action set; most algorithms leave the span, e.g.\ exponential weights.
\end{itemize}

\paragraph{First-order Query Complexity of Minimax Optimization}
Existing lower bounds \citep{IbrahimAGM20, Ouyang:2021aa} for minmax
optimization are known for finite-dimensional Euclidean spaces where
$\cX, \cY$ are either Euclidean balls or the whole space. Both these cases benefit from the property that the action sets are invariant under rotation, allowing for a step-by-step reduction to lower dimensional instances, as shown in the seminal work of \cite{nemirovsky1991on-optimality-o, nemirovsky1992information-bas}.
The notable exception of \cite{daskalakis2021the-complexity-} also does not cover the case we study: they consider the computational hardness of non-convex non-concave optimisation.

In the unconstrained case with curvature, that is, if $\mathcal F$ is the set of strongly convex, strongly concave functions with marginal condition numbers $\kappa_x$ and $\kappa_y$, the query complexity is settled to be $\smash{\Theta(\sqrt{\vphantom{h}\kappa_x \kappa_y} \log( 1/\eps) )}$; see \cite{IbrahimAGM20, zhang2022on-lower-iterat} for lower bounds and \cite{Lin2020NearOptimalAF} and references therein for upper bounds.

 For the constrained case, \cite{Ouyang:2021aa} provide lower bounds for bilinear saddle-point problems. They establish a query complexity of order $\Omega( L_f D_{\cX} D_{\cY} / \eps)$, where $D_{\cX}$  and $D_{\cY}$ are the respective diameters of the constraint sets $\cX$ and $\cY$, and $L_f$ measures the Lipschitz regularity of $f$. However, their techniques rely crucially on rotational invariance and sequentially adapted constraint sets, and they ask the open question of whether similar lower bounds can also be shown for fixed constraint sets.

\paragraph{Other Query Models}
The query complexity of minimax optimization has also been studied under different feedback models. Recall that we focus on the bilinear case, where $f: p, q\mapsto  p^\top M q$ for some game matrix~$M$. \cite{fearnley2015learning-equili} study the query complexity of (approximate) Nash equilibria of the game, i.e., (approximate) saddle points of $f$, under a query model where the learner chooses single entries of the matrix to observe. They show, among other results, that in order to compute an $\eps$-equilibrium in $K\times K$ zero-sum games, the number of entries queried needs to be at least $\Omega (K \log K)$ when $\eps = \mathcal O ( 1 / \log K)$.
\cite{hazan2016the-computation} consider a setting in which the learner
observes a best reponse to their query $(p,q)$, i.e., $i^\star \in
\argmin (M q)_i$ and $j^\star \in \argmin -(M^\top p)_j$. In that
setting, they show that to compute a $1/4$-equilibrium, at least
$\smash{\Omega(\sqrt K / (\log K)^2)}$ best-response queries are necessary.
Since a first-order query brings strictly more information than either a
best-response query or an entry query, these lower bounds do not have
direct consequences for our query model. 
Both of
these references prove their lower bounds by building hard matrices with
entries in a known finite set (e.g., $\{0, 1\}$).

\paragraph{Other Lower Bounds in Saddle Point and Equilibria Computation}
For multiplayer games, strong lower bounds on the computational complexity of exact Nash equilibria have been uncovered:  PPAD-hardness for computing the Nash-equilibrium in a general game, see \cite{ DBLP:journals/eccc/ChenDT06, 10.1145/1516512.1516516, article}, or in a non-convex-concave zero-sum game \cite{daskalakis2021the-complexity-}. Regarding query complexity, \cite{babichenko2016query-complexit, hart2018query} study the single-entry query complexity for approximate correlated equilibria and Nash equilibria of games with many players. As in other references mentioned, all these bounds are built on matrices with entries in a finite set.
Due to their intrinsically combinatorial nature, these techniques are less common for numerical algorithms. One of the classical methods for sample complexity lower bounds  was introduced by \cite{Nesterov2014IntroductoryLO} for first-order optimization. It was successfully extended to saddle-point problems \cite{zhang2022on-lower-iterat}  and recently also to unconstrained zero-sum games by \cite{IbrahimAGM20}. However, these techniques rely on the assumption that the next iterate is chosen in the span of the previous oracle information (cf.\ Definition~1 in \citep{IbrahimAGM20} or Assumption~2.1.4 in \citep{Nesterov2014IntroductoryLO}). Moreover, the game constructed by \cite{IbrahimAGM20} for a Nesterov-style lower bound results in matrices with entries from a finite alphabet.
For the same reason, no proof technique inspired by lower bounds via Rademacher random variables, as in \cite{ORABONA201850}, can ever work.

\subsubsection{Other Related Work}

\paragraph{Upper Bounds for Finite Action Zero-Sum Games}
Early on, \cite{brown1951iterative-solut, robinson1951an-iterative-me}
show that under best-response dynamics, the average plays converge to an
equilibrium. The connection to regret bounds was established by
\cite{freund1999adaptive}, who deduce a $\mathcal O(\log(K) /\eps^2)$ query complexity, by means of a construction akin to online-to-batch conversion.
More recently, \cite{daskalakis2011near-optimal-no,
daskalakis2015near-optimal-no} obtained the first fast rates of order
$\mathcal O ( c(K) / \eps)$, through an ingenious learning mechanism.
\cite{rakhlin2013optimization} rediscovered the Optimistic Online Mirror
Descent algorithm \citep{popov1980a-modification-}, and pioneered
optimistic online learning regret bounds; they yield to this the day the
fastest rates of $\mathcal O ( \log(K) / \eps)$ and the simplest algorithm.
The instance-dependent linear-rate upper bounds of order $\lambda(M) \ln
(1/\epsilon)$ by \cite{10.5555/1619995.1620009, Wei2020LinearLC} are
essentially incomparable to the worst-case $\mathcal O ( \log(K) / \eps)$; they are superior when $\lambda(M) \ll K$, and vacuous in the typical case $\lambda(M) \approx K$ (see the discussion below Theorem~\ref{thm:upperbd}).

Our focus is the quality of the inferred saddle point (this optimization perspective is sometimes called \emph{pure exploration}). Part of the literature regards the queries as actual moves made, and consequently prioritizes other objectives. In particular, many focus on studying \emph{uncoupled dynamics}, that is, sequences of actions that separate the observations of the individual $p$ and $q$ players without allowing communication between the players. Another recent theme of interest is the last-iterate convergence of such dynamics, see e.g., \cite{Hsieh2021AdaptiveLI}.
\paragraph{Upper Bounds for Minmax Optimization}

The vast majority of the known upper bounds in minimax optimization (and
more generally in variational inequality problems), outside of
finite-action zero-sum games mentioned above, concern unconstrained
settings \citep{azizian2020tight, golowich2020iterate, mokhtari2020convergence-rat}.
The literature on constrained settings is more sparse: \cite{cai2022tight} recently proved the rate of convergence of the projected extra-gradient method \citep{korpelevich1976extragradient}. \cite{yang2022solving} adapt interior methods to handle constraints. \cite{gidel2017frank} propose a fast algorithm when either the action sets are strongly convex, or the objective is strongly-convex strongly-concave. 
Recently, non-convex-concave settings have also attracted attention, see \cite{Lin2020NearOptimalAF} and references therein. 

\paragraph{Upper Bounds for General-Sum Games and Correlated Equilibria}
For multiplayer games, uncoupled dynamics do not converge (in any sense) to Nash equilibria. However, regret-based procedures can find correlated and coarse correlated equilibria \citep{Cesa-Bianchi:2006uo, stoltz2007learning-correl, piliouras2022beyond}. In particular, internal regret guarantees for individual players in general-sum multiplayer games were recently improved from the generic $\mathcal O(1 / \eps^2 )$ for any sequences of losses, to $O( 1 / \eps^{4/3})$ in \citep{syrgkanis2015fast}, to $O(1 / \eps^{6 / 5})$ in \citep{Chen2020HedgingIG} and to $O(\log( 1/\eps) / \eps )$ in \citep{daskalakis2021near-optimal-no, farina2022near-optimal-no} (we omit the dependence on the number of actions and players).
Polynomial-time methods for efficient computation of exact correlated equilibria are designed by \cite{jiang2011polynomial-time, papadimitriou2008computing-corre}, using their formulation as solutions to a linear program.

\section{Setting and Notation}
\paragraph{General Notation}
Given a set of numbers $A \subset \R$, we denote by $\mathcal M_K(A)$ the set of $K \times K $ matrices with entries in $A$; we mainly consider the bounded-entries class $\mathcal M_K([-1, 1])$.
A $K \times K$ zero-sum game between a minimizing $p$-player (or row player) and a maximizing $q$-player (or column player) is represented by a matrix $M \in \mathcal M_K(\R)$; we restrict our attention to square $K \times K$ games.
For any pair of plays $(p, q)$, the expected loss vector of the $p$-player (resp.\ $q$-player) is $Mq \in \R^K$ (resp.\ $-M^\top p \in \R^K$). The suboptimality gap at $(p, q)$ is 
\begin{equation*}
  g(M, p, q) = \max_{j \in [K]} (M^\top p)_j  - \min_{i \in [K]} (M q)_i \, . 
\end{equation*}
The gap $g(M, p, q)$ is non-negative for any $M, p$ and $q$. The pair of plays $(p, q)$ is said to be an $\eps$-Nash equilibrium if $g(M, p, q) \leq 2\eps$; Nash equilibria are $0$-Nash equilibria. Most games we will consider will possess a unique Nash equilibrium $p, q$, and this equilibrium will be fully supported, i.e., all components of $p$ and $q$ are positive. Fully supported equilibria in finite games are also equalizing strategies, meaning that the loss vectors are then isotropic; in finite zero-sum games, this entails that $M^\top p = Mq = v \mathbf 1$, where $v$ is the game-value.
We follow the notation convention to abbreviate $\min(x,v) = x\wedge v$ and to hide polylogarithmic terms by $\tilde \Omega$ or $\tilde O$. 

\paragraph{First-Order Query Model and Objectives}
The first-order query model is an interaction protocol between a learner and a game matrix, defined as follows. Fix a time horizon $T \in \mathbb N$, and a set of candidate matrices $\mathcal M \subseteq \mathcal M_K(\R)$. Over the course of $T$ rounds, the learner sequentially picks (queries) pairs of plays $(p_t, q_t)$ and observes the expected losses $(Mq_t, -M^\top p_t)$, where the matrix $M \in \mathcal M$ is fixed and unknown to the learner.

We examine the number of interactions necessary to compute an approximate equilibrium. At the end of the $T$ rounds, the learner recommends a pair $(p, q)$. Given a fixed gap level $\eps \geq 0$ and a set of candidate matrices $\mathcal M$, we say a strategy achieves query complexity $T(\eps; \mathcal M)$ if for any matrix $M \in \mathcal M$, the strategy outputs a pair $p, q$ such that $g(M, p, q) \leq 2\eps$ when $T \geq T(\eps; \mathcal M)$. (The dependence on $\mathcal M$ is omitted when clear from context.) We also occasionally refer to the query complexity of recovering the full matrix, which is a different task but in the same query model. In this case, the learner recommends a full matrix and achieves query complexity $T$ if the matrix guess is correct for any true matrix $M \in \mathcal M$.

The topic of this paper is the study of the optimal query complexity of finding $\eps$-Nash equilibria in the first-order model, for the set of matrices with bounded entries $\mathcal M_K([-1, 1])$.
 
\begin{figure}[h]
  \centering

  \begin{tikzpicture}[discontinuity/.style={decoration={show path construction,lineto code={%
    \path (\tikzinputsegmentfirst) -- (\tikzinputsegmentlast) coordinate[pos=.89] (mid);%
    \draw[ultra thick] (\tikzinputsegmentfirst) -- ([yshift=-6pt]mid) -- ++(-3pt,3pt) -- ++(6pt,3pt) -- ++(-3pt,3pt) -- (\tikzinputsegmentlast);%
  }}}]
\pgfplotsset{
              compat=1.12
}
\pgfplotsset{scale = 0.9}
\def \ymax {20}
\def \c {4}
 \def \upperBoundDot{(0.29957,20)(0.30877,19.3056)(0.31796,18.6513)(0.32715,18.0338)(0.33634,17.45)}
 \def \upperBound {(0.33634,17.45)(0.34554,16.8973)(0.35473,16.3733)(0.36392,15.8757)(0.37311,15.4026)(0.3823,14.9523)(0.3915,14.5232)(0.40069,14.1137)(0.40988,13.7226)(0.41907,13.3486)(0.42827,12.9907)(0.43746,12.6479)(0.44665,12.3192)(0.45584,12.0037)(0.46504,11.7007)(0.47423,11.4094)(0.48342,11.1292)(0.49261,10.8595)(0.50181,10.5996)}
\def \lowerBound{(0.000125,17)(0.0026735,9.4624)(0.0052219,7.8148)(0.0077704,6.8366)(0.010319,6.1386)(0.012867,5.5954)(0.015416,5.1507)(0.017964,4.7742)(0.020513,4.4477)(0.023061,4.1595)(0.02561,3.9015)(0.028158,3.6681)(0.030707,3.4548)(0.033255,3.2586)(0.035804,3.0769)(0.038352,2.9077)(0.040901,2.7494)(0.043449,2.6006)(0.045997,2.4603)(0.048546,2.3276)(0.051094,2.2017)(0.053643,2.0819)(0.056191,1.9677)(0.05874,1.8585)(0.061288,1.754)(0.063837,1.6538)(0.066385,1.5574)(0.068934,1.4647)(0.071482,1.3754)(0.074031,1.2892)(0.076579,1.2059)(0.079128,1.1253)(0.081676,1.0473)(0.084224,1)}
\newcommand{\myxlist}{0,1,2,\c * 1.3, \ymax -3, \ymax }
\begin{groupplot}[
group style={
    group name=my fancy plots,
    group size=2 by 2,
    xticklabels at=edge bottom,
    horizontal sep=0pt
}    ]

\nextgroupplot[
ymin=0, ymax={\ymax+1},
ytick/.expanded={ \myxlist },
yticklabels={0, 1, 2, $c \log K$,$\frac{K}{2}-1$,$K$},
     axis line style = {ultra thick, line cap=rect},
label style={font=\small},
tick label style={font=\small}, 
height=8cm,
  xlabel=$\epsilon$,
    ylabel={Queries $T(\epsilon) $},
xmin=0,xmax=0.63,
xtick={0,0.084,0.29957,0.5},
 xticklabels={0, $c/K^4$,$c \frac{\log K}{K}$,$\frac{1}{2}$},
axis x line*=box,
axis y line*=left,
every outer y axis line/.append style={discontinuity},
width=11cm]

   \addplot [color=red, smooth, ultra thick] plot coordinates \lowerBound;
   
     \addplot[smooth,color=red,ultra thick,dotted]
    plot coordinates {
        (0.6,1)
        (0.95,1)
    };
    
           \addplot[color=red!80!black!20,ultra thick]
    plot coordinates {
          (0.084,1)
           (0.62,1)                    };
              \addplot[color=red,ultra thick,dotted]
    plot coordinates {
          (0.084,1)
           (0.62,1)                    };
           
           \draw (0.025,5.2) -- ++ (70:0.4) node[above, text width=width("by Theorem 17")] {$\tilde\Omega\left(\log \frac{1}{K\epsilon}\right)$ by Theorem~\ref{thm:approx-lower-bound}};
            \draw  (0.41,13.7226) -- ++ (75:0.4) node[above,text width=width("Sridharan (2013) ")] {$O\left( \frac{\log K}{\epsilon}\right)$ \citet{rakhlin2013optimization}};
\draw  (0.6, 2) -- ++ (160:0.2) node[above,text width=width(" Daskalakis (2009) ")] {$T(\epsilon) = 1$  \citet{daskalakis2009note}};
    
      \addplot[color=blue!80!black!20,ultra thick, line cap=round]
    plot coordinates {
           (0.29957,20)(0,20)                   };

         \addplot[color=blue,ultra thick,dotted]
    plot coordinates {
           (0.29957,20)(0,20)                   };

    \addplot [color=blue, smooth, ultra thick, line cap=round] plot coordinates \upperBound;
     \addplot[color=blue,ultra thick]
    plot coordinates {
          (0.50181,10.5996)
          (0.50181,2)
         (0.62,2)
                };
                      \addplot[smooth,color=blue,ultra thick,dotted]
    plot coordinates {
        (0.6,2)
        (0.95,2)
    };
      \addplot [smooth,color=blue,ultra thick,dotted] plot coordinates \upperBoundDot;
 
\nextgroupplot[
ymin=0, ymax={\ymax+1},
ytick/.expanded={ \myxlist },
yticklabels={0, 1, 2, $c \log K$,$\frac{K}{2}-1$,$K$},
ytickmax=\ymax,
axis line style = {ultra thick, line cap=rect},
label style={font=\small},
tick label style={font=\small}, 
height=8cm,
xmin=0.8,xmax=1,
xtick={0.95},
xticklabels={$1-\frac{1}{K}$},
axis y line*=right,
every outer y axis line/.append style={discontinuity},
axis x line=box,
axis x discontinuity=crunch,
width=4cm
]
   \addplot[smooth,color=blue,ultra thick,dotted]
    plot coordinates {
        (0.7,2)
        (0.9,2)
    };

\addplot[color=blue,ultra thick]
    plot coordinates {
        (0.87,2)
        (0.95,2)(0.95,0.1)(1,0.1)
    };
    
       \addplot[smooth,color=red,ultra thick,dotted]
    plot coordinates {
        (0.7,1)
        (0.9,1)
    };
    \addplot[color=red!80!black!20,ultra thick]
    plot coordinates {
        (0.87,1)
        (0.95,1) 
    };

\addplot[color=red,ultra thick,dotted]
    plot coordinates {
        (0.87,1)
        (0.95,1)(0.95,0.1)(1,0.1)
    };

     \draw (0.96, 0.4) -- (0.96,7) node[above, xshift=-.5cm, text width=width("Prop. 2  xx")] {$T(\epsilon) = 0$ Prop.~\ref{prop:zeroQuery}};

        \end{groupplot}
\end{tikzpicture}
  \caption{\footnotesize Upper (blue) and lower (red) bounds on the first-order query complexity of computing an $\eps$-equilibrium for $K \times K$ matrix games with entries in $[-1, 1]$. See Sections~\ref{sec:upper-bounds} and~\ref{sec:lowerbounds} for details.}
  \label{fig:query-comp-graph}
\end{figure}
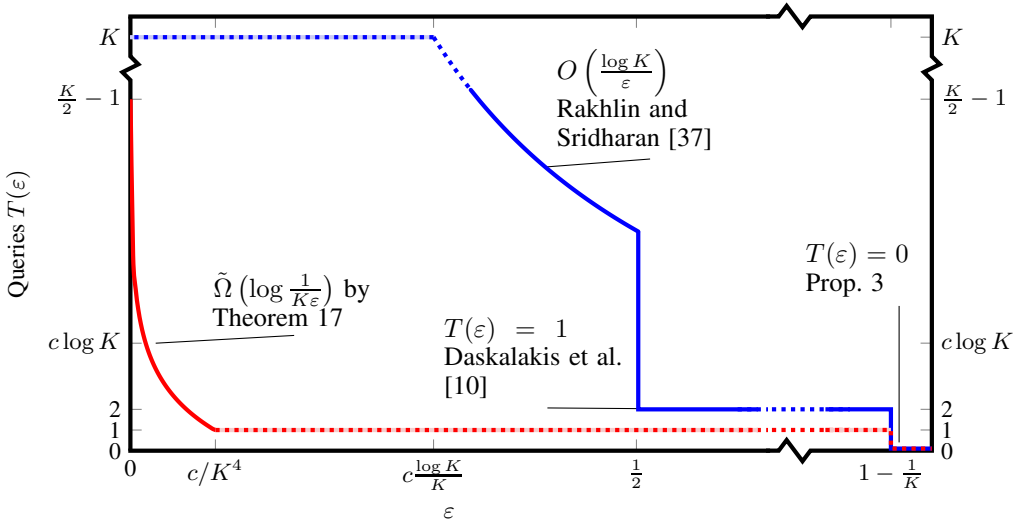

\section{Upper Bounds}\label{sec:upper-bounds}
In this section we collect upper bounds for the first-order query complexity of approximate Nash equilibria. These bounds are either well-established, trivial or exploit in a miraculously striking fashion the difference between learning saddle points for matrix games taking entries in a countable (e.g.,  $\mathbb Q$) or uncountable (e.g., $[-1, 1]$) sets of values.
\subsection{Query Complexity over $\mathcal M_K([-1, 1])$: Regret and Elementary Strategies}
We list some well-known results from the literature and state them in terms of first-order query complexity. The results described here are compiled in Figure~\ref{fig:query-comp-graph}. We use the notation $u_{n:m}$ to denote the family $(u_k)_{n \leq k \leq m}$.

\paragraph{$O(\epsilon^{-1} \log K)$ Queries from Optimistic Online Learning}
The current best upper bounds on the query complexity of $\eps$-equilibria for zero-sum games are derived from online learning methods. As is well-known \citep[Chapter 7]{Cesa-Bianchi:2006uo}, if we denote by $\widehat p_T$ (resp.\ $\widehat q_T$) the average of $p_{1:T}$ (resp.\ $q_{1:T}$) then 
\begin{multline*}
  T g(M, \widehat p_T, \widehat q_T) 
  = T \max_{j \in [K]} (M^\top \widehat p_T)_j  - T \min_{i \in [K]} (M \widehat q_T)_i  \\
  = \sum_{t = 1}^T \langle p_t,  Mq_t\rangle - T \min_{i \in [K]} (M \widehat q_T)_i 
  + \sum_{t = 1}^T \langle q_t, - M^\top p_t \rangle - T \min_{i \in [K]} (-M^\top \widehat p_T)_j \\
  \leq \sum_{t = 1}^T \langle p_t,  Mq_t\rangle - \sum_{t = 1}^T \min_{i \in [K]} (M  q_t)_i 
  + \sum_{t = 1}^T \langle q_t, - M^\top p_t \rangle - \sum_{t=1}^T \min_{i \in [K]} (-M^\top  p_t)_j
  \, . 
\end{multline*}
This last term is exactly the sum of the regrets suffered by each player on their respective losses: the gap of the average plays is smaller than the sum of the average regrets over $T$ rounds. This relationship between regret and gap provides a fruitful way to upper bound the query complexity of computing equilibiria. Specifically, \cite{rakhlin2013optimization} observed that if players follow the Optimistic Mirror Descent strategy then the sum of the regrets stays smaller than $\mathcal O (\log K)$; see the paragraph following Proposition~6 in the mentioned reference. 
\begin{theorem}[Consequence of \citet{rakhlin2013optimization}]\label{thm:upperbd}
  There exists an absolute constant $c>0$ such that the first-order query complexity over $\mathcal M_K([-1, 1])$ is
  \begin{equation*}
    T(\eps) \leq \Big(c \frac{\log K }{\eps} \Big) \wedge K \, .
  \end{equation*}
\end{theorem}
A wide stream of literature leverages the connection between regret and equilibria,  and study the dynamics of pairs of learning algorithms, including the two principal flavors of optimistic online algorithms, Optimistic Mirror Descent and Optimistic Follow-the-Regularized-Leader, cf.\ Section~\ref{sec:related-work}.

\paragraph{Instance-Dependent Rates of Convergence}
Many existing bounds show rates of convergence exponential in $T$, with instance-dependent constants\footnote{More precisely: a parameter similar to a condition number of the game matrix}, see, e.g., \cite{10.5555/1619995.1620009,Wei2020LinearLC}. However, we show in Example~\ref{CdeltaLarge} that for a zero-sum two-player bilinear game, these constants can be large, even on simple game matrices. If the constants are too large, the consequences on query complexity are vacuous, as they might require more than $K$ queries to obtain non-trivial ($\eps < 1$) equilibria. Although these results give a very interesting analysis beyond the worst-case, they are beyond the scope of this work.

\paragraph{Upper Bounds for Large $\eps$} 
We conclude this section with two elementary strategies that provide equilibria with large approximation values of $\eps\geq  1/ 2$. 
\begin{theorem}[Theorem. 3.1 in \cite{daskalakis2009note}]
  For any $\eps \geq 1/2$, the query complexity over $\mathcal M_K([-1, 1])$ of finding an $\eps$-equilibrium  is $T(\eps) \leq 2$. 
\end{theorem}
The following proposition is a well-known observation (see, e.g., \cite{fearnley2015learning-equili}) that having the learner select a pair of uniform plays always gives an approximate equilibrium.
\begin{proposition}
\label{prop:zeroQuery}
  For any $\eps \geq  1 - 1/K$, the query complexity over $\mathcal M_K([-1, 1])$ of finding an $\eps$-equilibrium is $T(\eps) = 0$.
\end{proposition}
\subsection{Recovering the Matrix: Upper and Lower Bounds}
\label{sec:doomedApproaches}
In this section we show that the first-order query complexity for recovering the full game matrix (a task harder than computing an equilibrium) is between $\Theta(1)$, independent of the number of actions, and $\Theta(K)$, depending on the set of candidate matrices. An important consequence of these results is that several known standard lower bound techniques cannot be applied to provide lower bounds on the easier task of learning equilibria.
\paragraph{Matrix Sets with Countable Alphabet}
Many existing methods rely on building hard matrices with entries in a finite alphabet: we prove this cannot lead to a lower bound.
\begin{theorem}
  \label{thm:DoomedApproachCoutable}
    Let $\mathcal A \subset \R$ be a countable set of at least two
    numbers. Then the first-order query complexity over $\mathcal
    M_K(\mathcal A)$ of recovering the full matrix is $1$.
  \end{theorem}
This result exploits infinite precision in the feedback, and as such does not provide a reasonable algorithm that the players would use should they know that the matrix belongs to $\mathcal M_K(\mathcal A)$. The intent of this statement is to show that any attempt at a lower bound that builds matrices with entries in a fixed countable set (and in particular with integer entries) is doomed to fail after one query.
\begin{remark}
  The proof of Theorem~\ref{thm:DoomedApproachCoutable} does not provide an explicit algorithm. However, if the learner knows that the entries of the matrix are in a finite set $\mathcal A \subset \Q$, an explicit strategy can be easily described. In this case, $\mathcal A$ is of the form $\{ a_1/r, \dots, a_n/r\}$ with $a_i\in \mathbb Z$ and $r \in \N$. Consider a base $b = \max(a_{1:n}) - \min (a_{1:n})$ and consider a query with $p \propto (b^{-1}, b^{-2}, \dots, b^{-K})$ (the choice for $q$ is irrelevant here). The learner can deduce the full matrix $M$ from the single observation $M^\top p$. Of course, this strategy has no practical interest as soon as either $\mathcal A$ or $K$ is moderately large, since it requires arbitrary precision in the outputs.
\end{remark}
\paragraph{Difficulty of Recovering the Matrix Exactly}
In the first-order query model, the learner receives $2K$ numbers every round, so in order to fully determine an
arbitrary $K \times K$ matrix one needs at least $K^2/(2K) = K/2$ rounds.
It turns out that this bound is not tight, and we may need exactly $K$
rounds in the worst case, because of redundancies between the $2K$ numbers we
observe per round:
\begin{theorem}\label{thm:mat_recovery}
  The first-order query complexity over $\mathcal M_K([-1, 1])$ of recovering the full matrix is $K$.
\end{theorem}
Intuitively, it is overkill to recover the matrix exactly to output an $\epsilon$-Nash equilibrium. In the next section, when $\eps \ll 1 / K^K$ (in particular when $\eps = 0$), we show that first-order queries do not provide useful information to find an $\eps$-equilibrium faster than it takes to reveal the whole matrix. 
\section{A New Lower Bound on the Query Complexity of Approximate Equilibria}
\label{sec:lowerbounds}
We switch perspective, attempting to make the life of any learner hard. For this, we need to think about responding to queries to keep the learner unwitting.
\subsection{Overview and Notation}
After $t$ time steps, queries $p_{1:t}$ and $q_{1:t}$ have been made and the outputs provided to the learner are $\smash{\ell_{1:t}^{(q)}}$ and $\smash{\ell_{1:t}^{(p)}}$. Let us denote the set of matrices with entries in $[-1, 1]$ that are consistent with the observations after $t$ time steps by
\begin{equation*}
  \mathcal E_t  = \big\{ M \in \mathcal M_{K}([-1, 1]) \; \big\vert \; \forall s \in [t] \, , \;
  M^\top p_s = -\ell_{s}^{(q)} \quad \text{and} \quad Mq_s = \ell_{s}^{(p)}
  \big\} \, ; 
\end{equation*}
we sometimes refer to this as the set of candidate matrices after $t$ rounds of observations. We omit the dependence on the sequence of queries $(p_s,  q_s)$ and outputs $\smash{(\ell_s^{(p)}, \ell_s^{(q)})}$ to reduce clutter, as it will always be clear from the context.

We say a sequence of matrices $M_{1:T}$ is adapted to the queries $p_{1:T}, q_{1:T}$, if it gives consistent outputs to the queries, i.e., if for all $s \leq t \leq T$, $M_t^\top p_s = M_s^\top p_s$ and $M_t q_s = M_s q_s$; in other words, $M_{1:T}$ is adapted if $M_{t+1} \in \mathcal E_t$ for all $t$.

Instead of defining directly the answers to the queries, we equivalently build a sequence of adapted matrices. Formally, at round $t+1$, given $\mathcal E_t$ and $(p_{t+1}, q_{t+1})$, we select a matrix $M_{t+1} \in \mathcal E_t$ and output $\smash{(\ell_{t+1}^{(p)}, \ell_{t+1}^{(q)})} = (M_{t+1} q_{t+1}, -M_{t+1}^\top p_{t+1})$.

Let us fix some common technical notation that we use in the proofs, to measure distances in matrix space.
For a matrix $M \in \mathcal M_K(\R)$, and for $r, s \in [1, \infty]$, we denote the operator norm of $M$ by
$
  \|M\|_{r, s} = \sup_{x \in \R^K : \, \|x\|_r = 1 } \|Mx\|_s \, . 
$
Recall that for $y, z \in \R^K$, we have $\|yz^\top \|_{r, s}  = \|y\|_s \|z\|_{r'}$, where $r'\in[1, \infty]$ is such that $1 / r + 1/r' = 1$. In particular $\|M\|_{1, \infty} = \max_{i, j\in [K]}|M_{i, j}|$, and $\|yz^\top\|_{1, \infty} = \|y\|_{\infty}\|z\|_{\infty}$. If $\mathcal F$ is a closed convex set subset of $\R^K$, we denote by $\Proj_{\mathcal F}(x)$ the orthogonal projection of $x$ on $\mathcal F$.

\paragraph{Common Structure of the Proofs}
Both proofs of Theorems~\ref{thm:exact_lower_bound} and~\ref{thm:approx-lower-bound} follow the same template:
\begin{itemize}[itemsep=0pt]
  \item[--] Assume an $\eps$-equilibrium is found after $T$ first-order queries.
  \item[--] Find a necessary condition on this equilibrium formulated as
  constraints on the observed losses and the sequence of queries.
  Precisely, we observe that under some initial assumptions on the game
  matrix, the all-ones vector $\mathbf 1$ needs to lie near the span of the outputs observed by the player, regardless of the queries and outputs.
  \item[--] Ensure that this necessary condition cannot be met by any pair of mixed actions $(p, q)$ by building an appropriate adapted sequence of matrices. In our case, we ensure that the span of the losses of the $q$-player stays far from the all-ones vector $\mathbf 1$.
\end{itemize}
This proof structure is a promising way to derive the exact query complexity of (approximate) Nash equilibria. We use it to provide the first non-trivial lower bound for this setting.

The next two sections implement this template to prove query complexity lower bounds for exact $\eps=0$ (Theorem~\ref{thm:exact_lower_bound}) and approximate $\eps>0$ (Theorem~\ref{thm:approx-lower-bound}) equilibria. Detailed proofs are in Appendix~\ref{app:proofsforlowerbounds}. While we could have deduced the final lower bound for exact equilibria directly from the approximate equilibria case, we describe the result separately to ease exposition. Indeed, the bound for the exact case contains the main ideas but is technically simpler.
\subsection{Proof: Exact Equilibrium Case}
\paragraph{Step I: Common Equilibria are in the Span of Past Queries}
We start by defining an initial set of candidate matrices with some special properties. These properties are used to ensure that the common equilibrium is fully mixed, and therefore an equalizing strategy for both players. 
\begin{definition}\label{def:A0}
  A set of matrices $B_0\subset \mathcal M_K([-1, 1])$ satisfies the assumptions $A(0)$ if 
  \begin{itemize}[noitemsep,topsep=0pt]
    \item[--] For all $M \in  B_0$, all equilibria of $M$ are fully mixed,
    \item[--] There exists a matrix $M \in B_0$ with non-zero value.
  \end{itemize}
\end{definition}
By Lemma~\ref{lem:ball} (App.~\ref{app:proofsforlowerbounds}), any ball centered at $(1/2)I_K$ with small enough radius satisfies $A(0)$, e.g.,
\begin{equation*}
  B_0 
  = \mathcal B_{\| \cdot\|_{1, \infty}}\Big(\frac{1}{2}I_K, \frac{1}{16K^2}\Big)
  = \Big\{ M \in \mathcal M_K(\R)  : \Big|M_{i, j} -  \frac{1}{2} \delta_{i = j}\Big| \leq \frac{1}{16K^2} \quad  \forall i, j \in [K] \Big\} \, .
\end{equation*}
The following lemma states that after $t$ first-order queries, if at least one matrix in $B_0$ is still a candidate matrix, then any common equilibrium to all candidate matrices must lie in the span of the queries. Since equilibria of matrices in $B_0$ are fully mixed, and are thus equalizing strategies, this further implies that the all-ones vector $\mathbf 1$ must lie in the span of the outputs to the queries.
\begin{lemma}\label{lem:p-in-span}
  Let $B_0$ be a set of matrices satisfying $A(0)$. Assume in addition that $\mathcal E_t \cap B_0 \neq \varnothing$ and that $\mathcal E_t \cap \mathcal M_K((-1, 1)) \neq \varnothing$. If there exists a common Nash equilibrium $(p, q)$ to all matrices in $\mathcal E_t$, then $p \in \Span(p_{1:t})$ and $q \in \Span(q_{1:t})$. 
\end{lemma}
%
%
\begin{corollary}\label{cor:v-in-span}
  Under the assumptions of Lemma~\ref{lem:p-in-span}, there exists a value $v \neq 0$ such that
  \begin{equation*}
    v \mathbf 1 
    \in \Span( \ell_{1:t}^{(q)}) 
    \cap \Span( \ell_{1:t}^{(p)}) \, . 
  \end{equation*}
\end{corollary}
\paragraph{Step II: Sequential construction}
Given Corollary~\ref{cor:v-in-span}, to prove a query complexity lower bound, it suffices to build the answers to the queries $p_{1:t}$ and $q_{1:t}$ in a way that ensures that
\begin{itemize}[noitemsep,topsep=0pt]
 \item[--] the vector $\mathbf 1$ never belongs to the span of the observations of (say) the $q$-player,
 \item[--] there is at least one remaining candidate matrix in $B_0$, i.e., $\mathcal E_t \cap B_0 \neq \varnothing$ \, . 
\end{itemize}
For technical reasons (cf.\ the assumptions of Lemma~\ref{lem:p-in-span}), we also need to make sure that there is enough space left in $\mathcal E_t$. We do so by keeping a candidate matrix $M_t \in \mathcal E_t \cap \mathcal M_K((-1,1))$, away from the border of the initial set of candidate matrices (i.e., with entries stricly between $-1$ and $1$).
\begin{lemma}\label{lem:construction-mt-exact}
  Fix a time horizon $T \leq K/2-1$. For any sequence of queries $p_{1:t}, q_{1:t}$, there exists a sequence of matrices $M_{1:T}$ in $\mathcal M_K((-1, 1))$ adapted to $p_{1:T}, q_{1:T}$ that defines losses $\smash{\ell_{1:T}^{(q)}}$ for which, for any $v \neq 0$,
  \begin{equation*}
    v \mathbf 1 \notin \Span(\ell_{1:T}^{(q)}) \, . 
  \end{equation*}
\end{lemma}

\paragraph{Conclusion}
Combining Lemmas~\ref{lem:p-in-span} and~\ref{lem:construction-mt-exact}, against any learning strategy, we have built a sequence of outputs for which there is no common equilibrium to all remaining candidate matrices after $T$ rounds, as long as $T \leq (K-3)/2$; this proves the following theorem.
\begin{theorem}\label{thm:exact_lower_bound}
  The first-order query complexity over $\mathcal M_K([-1, 1])$ of finding a Nash equilibrium is
  \begin{equation*}
    T(0) \geq K / 2 - 1 \, .
  \end{equation*}
\end{theorem}
The next section tackles the lower bound for approximate equilibria. Its proof follows the same template as for the exact case, although the technical complexity increases.
\subsection{Proof: Approximate Equilibria Case}
\paragraph{Step I: Common Equilibria Are Close to the Span of Queries}
In the following, we say a probability distribution over $[K]$ is $\delta$-supported if $p(i)\geq \delta$ for all $i$. We say a pair of distributions $(p, q)$ is $\delta$-supported if both $p$ and $q$ are $\delta$-supported. We start by defining a quantitative version of Definition~\ref{def:A0} for approximate equilibria. In order to retain the property that losses at $\eps$-equilibria stay close to an isotropic vector, we add a requirement on the support of the equilibria.
\begin{definition}\label{def:Aeps}
  A set of matrices $B_{\eps, \delta}\subset \mathcal M_K([-1, 1])$ satisfies the assumption $A(\eps, \delta)$ if
  \begin{itemize}[noitemsep,topsep=0pt]
    \item[--] For all $M \in  B_{\eps, \delta}$, all $\eps$-equilibria of $M$ are $\delta$-supported, 
    \item[--] There exists a matrix $M \in B_{\eps, \delta}$ with non-zero value.
  \end{itemize}
\end{definition}
For example, by Corollary~\ref{cor:open-fully-mixed-eps-ne}, the $\|\cdot\|_{1, \infty}$-ball centered at $(1/2) I_K$ and of radius $(1 / 16) K ^2$ satisfies this condition for any $\eps \leq 1 / (16K^2)$ and $\delta \leq  1/(2K)$.

The following proposition is a quantitative version of Lemma~\ref{lem:p-in-span} for approximate equilibria, in which we show that any common approximate equilibrium to all matrices in $\mathcal E_t$ needs to be close to the span of the queries.
\begin{lemma}\label{lem:p-in-span-appr}
  For $\eps, \delta > 0$, let $B_{\eps, \delta}$ be a set of matrices satisfying $A(\eps, \delta)$, cf.\ Definition~\ref{def:Aeps}. Assume that $\mathcal E_t \cap B_{\eps, \delta} \neq \varnothing$, and that the relative interior of $\mathcal E_t$ contains a ball of radius $r_t$ measured in $\|\cdot\|_{1, \infty}$-norm.
  If there exits a common $\eps$-Nash equilibrium $(p, q)$ sto all matrices in $\mathcal E_t$, then
  \begin{equation*}
    \| p - \Proj_{\Span(p_{1:t})}(p) \| \leq \frac{2\eps}{\delta r_t} 
    \quad \text{and } \quad 
    \| q - \Proj_{\Span(q_{1:t})}(q) \| \leq \frac{2\eps}{\delta r_t} \, .
  \end{equation*}
\end{lemma}
\begin{corollary}\label{cor:1closetospan}
  Under the assumptions of Lemma~\ref{lem:p-in-span-appr}, there exists a game matrix $M \in B_{\eps, \delta}$ such that the value $v \in \R$ of $M$ satisfies
  $
    \|v \mathbf 1 - \Proj_{\Span(\ell_{1:t}^{(q)})}(v \mathbf 1)\| \leq 4\sqrt K \eps  / (\delta r_t) \, .
  $
\end{corollary}

  \paragraph{Step II: Construction of Matrices}
  We design an adapted sequence of matrices that keeps $v \mathbf 1$ away from the span of the observed losses.
  \begin{lemma}\label{lem:construction-mt}
    Let $B$ be a closed ball in $\|\cdot \|_{1, \infty}$-norm of radius $r \leq 1/2$ contained in $\mathcal M_K([-1, 1])$.
    Fix a time horizon $T \leq K/2-1$. For any sequence of queries $(p_t, q_t)_{t \leq T}$, there exists a sequence of matrices $M_{1:T}$ in $B$ adapted to $(p_t, q_t)$ that defines losses $\ell_{1:T}^{(q)}$ for which, for any $v \geq 0$,
    \begin{equation*}
      \|v \mathbf 1 - \Proj_{\Span(\ell_{1:T}^{(q)})}(v \mathbf 1) \|^2
      \geq v^2 K \Big(\frac{r^2}{8KT^2} \Big)^{T+1} \, ,
    \end{equation*}
    and such that $\mathcal E_t$ contains a ball of radius $r/2$ in its relative interior.
  \end{lemma}
  The proof relies on a decomposition of the squared distance of $v\mathbf 1$ to the span of observed losses at time $t+1$ through a Pythagorean identity, relating it to the span at time $t$. By carefully choosing the new matrix $M_{t+1}$ as a function of the query $p_{t+1}, q_{t+1}$, we manage to ensure that the squared distance decreases only by a constant factor.
  \begin{remark}
    The radius of the ball inside $\mathcal E_t$ stays
    constant at $r /2$, even though $\mathcal E_t$ gets smaller as $t$
    increases. This is made possible by choosing $M_t$ very close to
    $M_{t-1}$ (roughly at a distance of $r / (KT)$ measured in
    $\|\cdot\|_{1, \infty}$), effectively ensuring that $M_t$ stays far
    from the boundary of $\mathcal M_K([-1, 1])$. 
  \end{remark}

  \paragraph{Step III: Conclusion}
  We now combine Corollary~\ref{cor:1closetospan} and Lemma~\ref{lem:construction-mt} to obtain a lower bound on the best achievable gap after $T$ rounds, which directly translates to a query complexity lower bound.
  \begin{theorem}\label{thm:approx-lower-bound}
    In the first-order query model on $\mathcal M_K([-1, 1])$, for any algorithm the worst-case gap after $T \leq (K-3)/2$ time steps is at least
    \begin{equation*}
      \eps \geq \frac{1}{2^{10} K^4} \Big(\frac{1}{2^{11/2} K^{5/2}T} \Big)^{T+1} \, . 
    \end{equation*}
    Therefore, the query complexity of finding an $\eps$-equilibrium for any $\eps \leq 1 / (e\,  2^{11} K^4)$ is at least 
    \begin{equation*}
      T(\eps) \geq \Big(\frac{- \log (2^{11} K^4 \eps ) }{\log( 2^{11/2} K^{5/2}) + \log(- \log (2^{11} K^4 \eps ))}  - 1 \Big)\wedge (K / 2 - 1) \, .
    \end{equation*}
  \end{theorem}
\subsection{Potential Improvement and Discussion}
There is still a wide disparity between upper and lower bounds on first-order
query complexity. The lower bound Theorem~\ref{thm:approx-lower-bound}
is most probably not tight, so let us discuss potential ways to improve
it.

We believe the most promising approach for improvement is to find a
different proxy for the gap. Our proof uses the distance to the span of
the losses $\|v \mathbf 1 - \Proj_{\Span(\smash{\ell_{1:t}^{(q)} })}\|^2$ for $v
\geq 1 / (2K)$, which is convenient because it is a distance, but we
need to introduce a strong restriction on the set of candidate matrices
$B_{\eps, \delta}$ to be able to relate it to the gap, namely the
restriction to a ball around $\half I_K$ of radius $O(1/K^2)$. In
order to make significant progress it is therefore essential to enlarge
the class of candidate matrices significantly compared to $B_{\eps,\delta}$, and therefore to modify the proxy for the gap.
\section{Discussion: Future Work and Conclusion}
We study the first-order query complexity of computing approximate Nash equilibria for two-player zero-sum matrix games. We review upper bounds coming from online learning, and discuss existing lower bounds for related problems including alternative query models. Taking stock, we arrive at the surprising state of affairs that for this fundamental problem no lower bounds are known. We then offer some explanation for this current state of affairs: the first-order query model is powerful enough to identify any matrix from a fixed countable set in a single query; this rules out many techniques. We then turn to lower bounds. We design an adaptive adversary that answers incoming learner queries in such a way that the remaining consistent matrices do not share a common Nash equilibrium for as long as possible. Our approach is based on a quantity serving as a ``potential function'': namely the distance of the all-ones vector to the span of the observations. We discuss in detail the result, future scope and limits of our proposed technique.

As can be seen in Figure~\ref{fig:query-comp-graph}, we are still far
from matching lower and upper bounds. The most intriguing possibility
for resolution would be if it were to turn out that the upper bounds
(i.e.\ algorithms) are improvable. Current upper bounds come from online learning regret bounds, with algorithms falling in the category of uncoupled dynamics. We would love to know if these algorithm templates are in fact optimal for query complexity.

To cycle back to our motivation of computing saddle points in general,
future work could attempt to extend our lower bound technique to
different constraint sets, to functions possibly exhibiting curvature,
and to instance-dependent rates. Another direction would be to generalize to multi-player games and investigate the query complexity of weaker solution concepts.

\paragraph{Limitations and Broader Impact} The main limit of this work is that it only partially resolves the first-order query complexity of approximate NE in finite games, calling for tighter analysis. Regarding broader impact, results are mainly theoretical and do not entail direct societal consequences.

\begin{ack}
  Sachs and Van Erven were supported by the Netherlands Organization for Scientific Research (NWO) under grant number VI.Vidi.192.095. Part of the research was performed while Hadiji was at the University of Amsterdam. During that time he was supported by the same grant number VI.Vidi.192.095.
\end{ack}

\bibliographystyle{plainnat}
\bibliography{swap-bib}

\newpage

\appendix

\section{Proofs of Lemmas~\ref{lem:p-in-span} and \ref{lem:construction-mt-exact}}

\begin{proof}[\textbf{Proof of Lemma~\ref{lem:p-in-span}}]
  Since $p, q$ is an equilibrium to at least one matrix in $B_0\cap \mathcal E_t$, it is fully supported; in particular $p$ is an equalizing strategy for any $M \in \mathcal E_t$ and $M^\top p = v_M \mathbf 1$ for some number $v_M \in \R$. Furthermore, the value $v_M$ is actually independent of $M \in \mathcal E_t$ since, $v_M = \langle M^\top p, q_1\rangle = \langle p, \ell_{1}^{(p)} \rangle $. Therefore, for any $M, M' \in \mathcal E_t$, we have 
  $
    (M - M')^\top p = \mathbf 0 \, .
  $
  Let us now define $\bar p = p - \Proj_{\Span(p_{1:t})}(p)$, and consider the direction $\bar p u_q^\top \in \mathcal M_K(\R)$ for some arbitrary non-zero vector $u_q$ orthogonal to $q_1, \dots, q_t$. Fix some matrix $M \in\mathcal E_t \cap \mathcal M_K((-1, 1)) $; a non-empty set by assumption. Then for $\alpha \in \R$ small enough, the matrix $M' = M + \alpha \bar p e_q^\top$ is still in $\mathcal E_t$, since $M$ is not on its border (all entries are away from $\{-1, 1\}$). Therefore, 
  \begin{equation*}
    \mathbf 0 = (M - M')^\top p = \alpha \langle p, \bar p \rangle e_q = \alpha \|\bar p\|^2 u_q \, ; 
  \end{equation*}
  implying that $\|\bar p\| = 0$, i.e.\ that $p \in \Span(p_{1:t})$; similar reasoning shows that $q \in \Span(q_{1:t})$.
\end{proof}

\begin{proof}[\textbf{Proof of Lemma~\ref{lem:construction-mt-exact}}] 
    We build the sequence $M_{1:T}$ incrementally by moving at step $t+1$ in directions chosen as a function of the new queries $p_{t+1}$ and $q_{t+1}$, ensuring by induction that $\mathbf 1 \notin \Span(\ell_{1:T}^{(q)})$.
    
    Initialize the sequence at $M_0 = (1/2)\, I_K$. Now for $t\geq 0$, let us assume that we have correctly built $M_{1:t}$, inducing losses such that $\mathbf 1 \notin \Span(\ell_{1:t}^{(q)})$, and let us define $M_{t+1}$. (Note that the initialization of the induction is valid with the convention that $\mathbf 1 \notin \{\mathbf 0\} = \Span(\varnothing)$.)
    
    If $p_{t+1}$ is in the span of $p_{1:t}$, then set $M_{t+1} = M_t$ and the induction holds, since the span of the losses is left unchanged. Otherwise set
    $
      M_{t+1} = M_t + \frac{\bar p_{t+1}}{\|\bar p_{t+1}\|^2} u_t^\top \, ,
    $
    where $\bar p_{t+1} = p_{t+1} - \Proj_{\Span(p_{1:t})}(p_{t+1})$, 
    and $u_t$ is a non-zero vector orthogonal to the vectors $q_{1:t}$, to $\ell_{1:t}^{(q)}$, to $\mathbf 1 $ and to $M_t^\top p_{t+1}$; the existence of such a $u_t$ is guaranteed since $2t + 2 < K $, ensuring that there is at least one dimension orthogonal to those $2t+2$ vector. By choosing the norm of $u_t$ to be small enough, we can make sure that $M_{t+1} \in \mathcal E_t \cap \mathcal M_K((-1, 1))$. Then, 
    $
      \ell_{t+1}^{(q)} = M_{t+1}^\top p_{t+1} = M_t^\top p_{t+1} + u_t
    $,  
    and $\mathbf 1$ is not in the span of $\ell_{1:(t+1)}^{(q)}$. Indeed, assume by contradiction that there exists real numbers $\alpha_{1:(t+1)}$ such that
    \[
      \mathbf 1 = \sum_{s = 1}^t \alpha_{s}\ell_s^{(q)} + \alpha_{t+1} \ell_{t+1}^{(q)} \, ,
    \]
    then $\alpha_{t+1} > 0$ since by induction $\mathbf 1 \notin \Span(\ell_{1:t}^{(q)})$. Then, taking dot products with $u_t$ on both sides of the equation above, we obtain $0 = \langle u_t, \mathbf 1\rangle = \langle u_t,  \alpha_{t+1} \ell_{t+1}^{(q)} \rangle = \alpha_{t+1}\|u_t\|^2$, leading to a contradiction. Therefore $\mathbf 1 \notin \Span( \ell^{(q)}_{1:t})$ at all times $t \leq T$.
\end{proof}

 \section{Details for Upper bounds}
 \label{UpperBounds}
\subsection{Instance-Dependent Query Complexity}
 The following example shows that the bounds in \cite{10.5555/1619995.1620009} can be vacuous for our setting. Since, as noted by the authors, Theorem~5  by \cite{Wei2020LinearLC}  is equivalent to the results in  \cite{10.5555/1619995.1620009}, the following example is only with respect to the latter.
 \begin{example}
    \label{CdeltaLarge}
    For this example, we use the notation from  \cite{10.5555/1619995.1620009}. Consider the game matrix $M = I_K$. Observe that $M$ has a unique Nash-equilibrium, hence, by considering  $p = \frac{1}{K} \one$ and $q = \delta_i$ it can be seen that $\delta(M)$ is at most $\sqrt{\frac{K}{K-1}}\frac{1}{K} \approx \frac{1}{K}$. The bound is defined with respect to the condition number $\kappa(M) = \sqrt{\lambda_{\max} M^\top M}/ \delta(M)$, where $\lambda_{\max}(M)$ denotes the largest eigenvalue of $M$. Hence, for our example, $\kappa(M)$ is at least of the order of $K$, which makes a bound of $\kappa(M) \log \frac{1}{\epsilon}$ for this specific example vacuous. Note that for other game matrices, the results give valuable insights into guarantees beyond the worst-case. 
 \end{example}

\subsection{Constant Queries for Games from a Finite Alphabet}
\begin{proof}[\textbf{Proof}]
  Let $F$ be the smallest subfield of $\R$ that contains $\mathcal A$, then $F$ is countable and, $\R$ can be seen as an infinite-dimensional vector space over $F$. Recall that a family of real numbers $(x_1, \dots, x_n)$ (each seen as a vector over the field $F$) is linearly independent if for any $\lambda_1, \dots, \lambda_n \in F$ we have $\lambda_1 x_1 + \dots + \lambda_n x_n = 0$ if and only if $\lambda_1 = \dots = \lambda_n = 0$.

  We claim that if a player, say the $q$-player, queries an action such that the components $q_1, \dots, q_K$ form a linearly independent family, then they can compute the whole matrix $M$ with just one observation. Indeed, if $M, M' \in \mathcal M_K(\mathcal A)$, yield the same output after one query, then for any $i \in [K]$:
  \begin{equation*}
    \sum_{i = 1}^K (M_{i, j} - M'_{i, j}) q_i = 0. 
  \end{equation*}
  Therefore, by the independence of $(q_i)_{i \in [K]}$, this implies that $M_{i, j} = M'_{i, j}$ for any $(i, j)$. In other words, no two different matrices can give the same output after one query under $q$.

  We are now left to show that there exists a play $(q_1, \dots, q_K)$ with coordinates forming a linearly independent family. We prove this using the probabilistic method. Consider a sequence of random variables $(U_1, \dots, U_K)$ i.i.d. and uniformly distributed over $[0, 1]$. Then with probability $1$, the $U_i$ are independent over $F$. Indeed, we can upper bound the probability that they are dependent by a union bound and use of the tower rule as
  \begin{multline*}
    \P\big[ \exists i\in[K] \quad \mbox{s.t.} \quad U_i \in \Span_{F}\{U_j \,|\, j \neq i \}\big]
     \leq \sum_{i = 1}^K \P\big[ U_i \in \Span_{F}\{U_j \mid j \neq i \}\big]  \\
     = \sum_{i = 1}^K \E \big[\P\big[ U_i \in \Span_{F}\{U_j | j \neq i \} \mid \{U_j \mid j\neq i\}  \big] \big] = 0 \, . 
  \end{multline*}
  The last equality holds because, for any $i \in [K]$, conditionally on $\{U_j \mid j \neq i\}$, the span of $\{U_j \mid j \neq i\}$ is a countable set, therefore the probability that $U_i$ belongs to that set is null. This concludes the proof.
\end{proof}
\section{First-Order Query Complexity of Recovering the Game Matrix}
\begin{proof}[\textbf{Proof of Theorem~\ref{thm:mat_recovery}}]
  Clearly, we can fully reconstruct $M$ from the queries $p_t = q_t =
  e_t$ for $t=1,\ldots,K$, where $e_t$ denotes the standard basis vector
  in direction $t$. It turns out that this is optimal.

  To show this, note that each query $(p,q)$ provides us with
  constraints
  \begin{align*}
    p^\top M &= a,
    &
    M q &= b,
  \end{align*}
  for some loss vectors $a$ and $b$. These may equivalently be expressed
  as linear constraints in the Hilbert space of matrices $A \in
  \reals^{K \times K}$, with inner product $\ip{A,B} = \Tr(A^\top B)$:
  \begin{align*}
    \ip{M^\top, e_i p^\top} &= a_i  \qquad (i=1,\ldots,K),\\
    \ip{M^\top, q e_j^\top} &= b_j  \qquad (j=1,\ldots,K).
  \end{align*}
  Among these $2K$ constraints on $M$, there is (at least) one redundant
  constraint, because there always exist numbers
  $\lambda_1,\ldots,\lambda_K$ and $\gamma_1,\ldots,\gamma_K$, at least
  one of which is nonzero, such that
  \[
    \sum_{i=1}^K \lambda_i e_i p^\top + \sum_{i=1}^K \gamma_j q e_j^\top
    = 0.
  \]
  Specifically, this holds for $\lambda_i = q_i$ and $\gamma_j = -p_j$.
  It follows that a query $(p_t,q_t)$ in round $t$ will provide at most
  $2(K-t) + 1$ new constraints on top of the queries from rounds
  $1,\ldots,t-1$. To see this, note that $p_t$ will have at least one
  constraint in common with $q_1,\ldots,q_{t-1}$ and $q_t$ will have at
  least one constraint in common with $p_1,\ldots,p_t$, so the total
  number of new constraints is at most $n_t := 2K - (t-1) - t = 2(K-t) +
  1$.

  We now provide the following scenario in which $M$ cannot be fully
  determined by strictly less than $K$ queries. In rounds
  $t=1,\ldots,K-1$, the answer to queries $(p_t,q_t)$ is always the two
  loss vectors $\ell_t^{(p)} = \ell_t^{(q)} = \half \ones$, which are
  compatible with the possibility that $M$ equals $\half \ones
  \ones^\top$, i.e.\ the matrix with all entries equal to $1/2$. We will
  show by induction that the dimension of the null-space (i.e.\ the
  number of unconstrained dimensions of $M$) of all constraints up to
  round $t$ is at least $(K-t)^2$. This is true for $t=0$, because the
  domain of $M$ has $K^2$ dimensions, and, whenever it is true for $t$,
  then for $t+1$ it is at least
  \[
    (K-t)^2 - n_{t+1}
      = (K-t)^2 - 2(K-t-1) - 1
      = (K-t)^2 - 2(K-t) + 1
      = (K-t-1)^2.
  \]
  Thus, after $K-1$ rounds, there remains at least one direction $V \in
  \reals^{K\times K}$ in this null space with $V \neq 0$. This means
  that the learner cannot distinguish the case $M=\half \ones
  \ones^\top$ from the case $M=\half \ones \ones^\top + \half V/\max_{i,j}
  |V_{i,j}|$, thus $K-1$ queries are not sufficient to fully
  determine~$M$.
\end{proof}
\section{Proofs and Technical Results for Section~\ref{sec:lowerbounds}} 
\label{app:proofsforlowerbounds}

\subsection{Approximate Equilibria}

\subsubsection{Proofs of Main Results}

\begin{proof}[{\bfseries Proof of Corollary~\ref{cor:1closetospan}}]
  If $(p, q)$ is a common $\eps$-NE to all matrices in $\mathcal E_t$, then we have for any game matrix $M \in \mathcal E_t$,
  \begin{equation*}
    M^\top (p - \bar p) \in \Span\big( \ell_{1:t}^{(q)}\big) \, 
  \end{equation*}
  and therefore for any $v > 0$,
  \begin{equation*}
    \| v \mathbf 1 - \Proj_{\Span(\ell_{1:t}^{(q)})}(v \mathbf 1)  \|
    \leq 
    \| v \mathbf 1 - M^\top (p - \bar p)  \|
    \leq 
    \| v \mathbf 1 - M^\top p\| + \|M^\top \bar p\|
    \leq \| v \mathbf 1 - M^\top p\| + \sqrt K \|\bar p\| \, .
  \end{equation*}
  Now as $\delta > 0$, any $M \in \mathcal B_{\eps, \delta}$ has a unique Nash equilibrium $p^\star, q^\star$, which is fully supported. Therefore, the value of $M$ is $v = p^\top M q^\star$, as $q^\star$ is a an equalizing strategy. Now, using \eqref{eq:range-at-eq} in the proof of Lemma~\ref{lem:p-in-span-appr}, which is valid for any matrix that admits $p, q$ as an $\eps$-NE,
  \begin{equation*}
    \|v \mathbf 1 - M^\top p \| \leq \sqrt K\|v \mathbf 1 - M^\top p \|_\infty \leq \frac{\sqrt K\eps}{\delta} \,.
  \end{equation*}
  \end{proof}
\begin{proof}[\textbf{Proof of Lemma~\ref{lem:p-in-span-appr}}]
  First note that $p$ and $q$ are $\delta$-supported, as $(p, q)$ is an $\eps$-equilibrium of at least one matrix in $\mathcal E_t \cap B(\eps, \delta)$. For $M \in \mathcal E_t$, pick any $j^\star \in \arg \min_{j \in [K]} (M^\top p)_j$. As $(p, q)$ is an $\eps$-equilibrium for $M$,
  \begin{equation*}
      \big(1 - q(j^\star)\big) \max_{j \in [K]} (M^\top p)_j +  q(j^\star) \min_{j \in [K]} (M^\top p)_j \geq \sum_{i = 1}^K q(j) (M^\top p)_j \geq \max_{j \in [K]} (M^\top p)_j - \eps \, ,
  \end{equation*}
  and therefore, dividing by $q(j^\star) \geq \delta$, 
  \begin{equation}\label{eq:range-at-eq}
      \max_{j \in [K]} (M^\top p)_j - \min_{j\in[K]} (M^\top p)_j
      \leq \frac{\eps}{q(j^\star)}
      \leq \frac{\eps}{\delta} \, .
  \end{equation}
  Hence, as $\min_{j \in [K]} (M^\top p)_j  \leq p^ \top Mq_1 \leq  \max_{j\in[K]} (M^\top p)_j$,
  \begin{equation*}
      \| M^\top p  - (p^\top Mq_1) \mathbf 1 \|_{\infty} \leq \frac{\eps}{\delta} \, .
  \end{equation*}
  Since $p^\top Mq_1 = p^\top \smash{\ell_{1}^{(p)}}$ has the same value for any $M \in \mathcal E_t$, we apply this inequality for any pair of matrices $M, M' \in \mathcal E_t$ to deduce that
  \begin{equation}\label{eq:Mtoppuniform}
          \|(M-M')^\top p\|_\infty \leq \frac{2\eps}{\delta} \, .
  \end{equation}
  Let us now instantiate this identity with some well-chosen $M$ and $M'$. Let $u \in \R^K$ be a vector orthogonal to $q_1, \dots, q_t$, such that $\|u\|_\infty = 1$ (such a $u$ exists as long as $q_1, \dots, q_t$ do not span the whole of $\R^K$). Next consider $M$ and $M'$ in $\mathcal E_t$ such that
  \begin{equation*}
      M-M' = \frac{r_t}{\|\bar p\|_{\infty}} \bar p u^\top \, .
  \end{equation*}
  Such $M$ and $M'$ are guaranteed to exist in $\mathcal E_t$ as
  $
      \|M - M'\|=  r_t \, . 
  $
  Then, applying \eqref{eq:Mtoppuniform}, we obtain
  \begin{equation*}
      \frac{2\eps}{\delta} 
      \geq \|(M-M')^\top p\|_\infty
      = \frac{r_t}{\|\bar p\|_\infty  } \|\bar p \|_2^2  \|u\|_\infty
      \geq r_t \|\bar p\|_2  \, ,
  \end{equation*}
  from which the claim follows. The same reasoning applies to obtain the bound on $q$.
  \end{proof}

\begin{proof}[\textbf{Proof of Lemma~\ref{lem:construction-mt}}]
  Denote by $\Pi_t$ the projection on the span of the observed losses $\ell_{1:t}^{(q)}$ for the $q$-player at time step $t$. We build the sequence $(M_t)$ incrementally by moving at step $t+1$ in directions chosen as a function of the new queries $p_{t+1}$ and $q_{t+1}$.

  We initialize the sequence at $M_0$ the center of $B$. Now for $t\geq 0$, if $p_{t+1}$ is in the span of $p_{1:t}$, then set $M_{t+1} = M_t$, otherwise set
  \begin{equation*}
    M_{t+1} = M_t + \frac{\bar p_{t+1}}{\|\bar p_{t+1}\|^2} u_t^\top \, ,
  \end{equation*}
  where $\bar p_{t+1} = p_{t+1} - \Proj_{\Span(p_{1:t})}(p_{t+1})$, 
  and $u_t$ is a non-zero vector orthogonal to the vectors $q_{1:t}$, to $\ell_{1:t}^{(q)}$, to $\mathbf 1$, and to $M_{t}^\top p_{t+1}$; the existence of such a $u_t$ is guaranteed since the only condition is that it is orthogonal to $2t + 2 < K$ vectors. (Note that $u_t$ does not depend on the value of $v$.) We set the norm of $u_t$ at a later stage of the proof. Then, 
  \begin{equation*}
    \ell_{t+1}^{(q)} = M_{t+1}^\top p_{t+1} = M_t^\top p_{t+1} + u_t\, . 
  \end{equation*}
  Then for any $v \in \R$, the squared distance from the vector $v\mathbf 1$ to the space $\Span(\ell_{1:t+1}^{(q)})$ can be decomposed thanks to the Pythagorean equality, as the squared distance to the previous span minus the squared norm of the projection on the new orthogonal direction $\smash{\ell_{t+1}^{(q)} - \Pi_t(\ell_{t+1}^{(q)})}$:
  \begin{equation*}
      \|v \mathbf 1 - \Pi_{t+1} (v \mathbf 1) \|^2
      = \|v \mathbf 1 - \Pi_{t} (v \mathbf 1) \|^2 - 
      \underbrace{\frac{\langle v \mathbf 1, \, M_t^\top p_{t+1} + u_t -\Pi_t(M_t^\top p_{t+1} + u_t) \rangle^2}{\|M_t^\top p_{t+1} + u_t - \Pi_t(M_t^\top p_{t+1} + u_t)\|^2}}_{:= D_t} \,  .
  \end{equation*}
  Observe that for any vectors $a, b$, we have $\langle a - \Pi_t(a), b\rangle = \langle a - \Pi_t(a), b - \Pi_t(b)\rangle  = \langle a, b- \Pi_t(b)\rangle $, as $\Pi_t$ is an orthogonal projection on a linear subspace. Using this identity, as well as the orthogonality conditions used to define $u_t$ (precisely, that $\Pi_t(u_t) = \mathbf 0$, and that $u_t$ is orthogonal to $M_t^\top p_{t+1}$ and  $\Pi_t(M_t^\top p_{t+1})$), we obtain after applying Cauchy-Schwarz,
  \begin{align*}
    D_t &= \frac{\langle v \mathbf 1 - \Pi_t(v \mathbf 1), \, M_t^\top p_{t+1} + u_t  \rangle^2}{\|M_t^\top p_{t+1} + u_t - \Pi_t(M_t^\top p_{t+1} + u_t)\|^2} \\
    &=   \frac{\langle v \mathbf 1 - \Pi_t(v \mathbf 1), \, M_t^\top p_{t+1} + u_t  \rangle^2}{\|u_t\|^2 + \|M_t^\top p_{t+1} - \Pi_t(M_t^\top p_{t+1})\|^2} \\
    &=   \frac{\langle v \mathbf 1 - \Pi_t(v \mathbf 1), \, M_t^\top p_{t+1} - \Pi_t(M_t^\top p_{t+1})  \rangle^2}{\|u_t\|^2 + \|M_t^\top p_{t+1} - \Pi_t(M_t^\top p_{t+1})\|^2} \\
    & \leq \|v\mathbf 1 - \Pi_t(v \mathbf 1)\|^2 \frac{\|M_t^\top p_{t+1} - \Pi_t(M_t^\top p_{t+1})\|^2}{\|M_t^\top p_{t+1} - \Pi_t(M_t^\top p_{t+1})\|^2 + \|u_t \|^2} \\
    & \leq  \|v\mathbf 1 - \Pi_t(v \mathbf 1)\|^2 \frac{\|M_t^\top p_{t+1} \|^2}{\|M_t^\top p_{t+1} \|^2 + \|u_t \|^2} \, . 
  \end{align*}
  We also used the fact that projections reduce the norm, and the function $x \mapsto x / (x +1)$ is increasing on $(0, +\infty)$ to obtain the final inequality.

  Using this bound on $D_t$ to lower bound the distance to the span, we obtain
  \begin{equation}\label{eq:distance-decrease}
      \|(I_K - \Pi_{t+1})v \mathbf 1\|^2 \geq \|(I_K - \Pi_{t})v \mathbf 1\|^2
      \biggl( 1 - \frac{ 1 }{1 +  \|u_t\|^2 / \|M_t^\top p_{t+1}\|^2} \biggr) \, . 
  \end{equation}
  We are now left to choose the norm of $u_t$; the objective is to
  make it as big as possible under the constraint that the sequence
  $(M_t)$ stays in $B$. We set the norm of $u_t$ to be a constant multiple of $\|M_t^\top p_{t+1}\|$, i.e., 
  \begin{equation*}
    \|u_t\| = \sqrt \alpha \|M_t^\top p_{t+1}\| \, .
  \end{equation*}
  For $\alpha$ small enough, we can ensure that the whole sequence $(M_t)$ stays in the ball $B$, as
  \begin{multline*}
      \|M_t - M_0\|_{1, \infty}
      \leq \sum_{s=1}^t \frac{1}{\|\bar p_{s+1}\|^2} \|\bar p_{s+1} u_s^\top\|_{1, \infty}
      = \sum_{s=1}^t \frac{\|\bar p_{s+1}\|_\infty}{\|\bar p_{s+1}\|^2}  \|u_s\|_{\infty} 
      = \sum_{s=1}^t \frac{\|\bar p_{s+1}\|_\infty}{\|\bar p_{s+1}\|^2}  \|u_s\|_{\infty} \\
      \leq \sum_{s=1}^t \frac{\|u_s\|}{\|\bar p_{s+1}\|} 
      = \sqrt{\alpha} \sum_{t=1}^t \frac{\|M_t \bar p_{t+1}\|}{\|\bar p_{t+1}\|}
      \leq \sqrt{\alpha} \sqrt K t \, ,
  \end{multline*}
  where we used $\|M_t x\| \leq \|M_t\|_{2, 2}\|x\| \leq \sqrt K \|x\|$, and $M_t^\top  \bar p_{t+1} = M_t^\top p_{t+1}$ since $M_t \in \mathcal E_t$. Therefore by taking, 
  $
    \alpha = {(r / 2)^2} / {(KT^2)} \, ,
  $
  we ensure that $M_t \in B$ as $B$ contains a ball of radius $r$ centered at $M_0$. 
  Furthermore, this choice ensures that $\mathcal E_t$ contains the ball of radius $r/2$ centered at $M_t$ in its relative interior as for any matrix $U$ with norm less than $r / 2$, we have $\|M_t + U - M_0\| \leq \|M_t - M_0\| + \|U\| \leq r / 2$, so $M_t + U \in B \subset \mathcal M_K([0, 1])$.
  
  Plugging back the value of $\alpha$ into \eqref{eq:distance-decrease}, we get
  \begin{equation*}
      \|(I_K - \Pi_{t+1})v \mathbf 1\|^2
      \geq \|(I_K - \Pi_{t})v \mathbf 1\|^2 \frac{\alpha}{ 1 + \alpha } 
      \geq \|(I_K - \Pi_{t})v \mathbf 1\|^2 \frac{\alpha}{ 2 }. 
  \end{equation*}
  After $T$ time steps, this implies that for any $v \geq 0$, 
  \begin{equation*}
    \|(I_K - \Pi_T)v \mathbf 1\|^2
     \geq \|(I_K - \Pi_0)v \mathbf 1\|^2 \Big(\frac{\alpha}{2}\Big)^T
     = \|v \mathbf 1\|^2 \Big(\frac{r^2}{8KT^2} \Big)^{T+1} \, .
  \end{equation*}
  This is the claimed result.
\end{proof}

\begin{proof}[\textbf{Proof of Theorem~\ref{thm:approx-lower-bound}}]
  Define $B_{\eps, \delta} = \mathcal B_{\| \cdot\|_{1, \infty}}\big(\frac{1}{2}I_K, \frac{1}{16K^2}\big)$, which is a ball of radius $r = 1 / (16K^2)$. Given an algorithm, consider the sequence of matrices $M_{1:T}$ generated by Lemma~\ref{lem:construction-mt}, applied with $B = B_{\eps, \delta}$; we know in particular that $M_t \in B_{\eps, \delta}$, so $B_{\eps, \delta} \cap \mathcal E_t \neq \varnothing$.

  Assume now that the algorithm outputs a common $\eps$-equilibrium to all matrices in $\mathcal E_T$. The assumptions of Corollary~\ref{cor:1closetospan} hold, instantiated with $B_{\eps, \delta}$ and $r_t = 1 / (32 K^2)$. Therefore there exists $M \in B_{\eps, \delta}$, with value $v$ such that
  \begin{equation*}
    \|v \mathbf 1 - \Proj_{\Span(\ell_{1:t}^{(q)})}(v \mathbf 1)\| \leq \frac{4\sqrt K \eps }{\delta r_t}.
  \end{equation*}
  On the other hand, by Lemma~\ref{lem:construction-mt} also instantiated with $B = B_{\eps, \delta}$, we know that 
  \begin{equation*}
    \|v \mathbf 1 - \Proj_{\Span(\ell_{1:T}^{(q)})}(v \mathbf 1) \|
    \geq v \sqrt K \Big(\frac{r^2}{8KT^2} \Big)^{(T+1) / 2} \, .
  \end{equation*}
  This implies that 
  \begin{equation*}
    \eps \geq \frac{v \delta r_t}{4} \Big(\frac{r^2}{8KT^2} \Big)^{(T+1) / 2}. 
  \end{equation*}
  Finally, since $v$ is the value of a matrix $M \in B_{\eps, \delta}$,
  \begin{equation*}
    v = \min_{p \in \Delta_K} \max_{q \in \Delta_K} p^\top Mq 
    \geq \min_{p \in \Delta_K} \max_{q \in \Delta_K} \frac{p^\top q}{2} - K r \geq \frac{1}{2K} - \frac{K}{16K^2} 
    \geq \frac{1}{4K} \, . 
  \end{equation*}
  We conclude by replacing the constants by their values, $\delta = 1 / (2K)$, and $r_t = r  /2 = 1 / (32K^2)$.
  By Lemma~\ref{lem:invert-query} in Appendix~\ref{app:proofsforlowerbounds}, we deduce that for $(2\eps) \leq 1 / (e \, 2^{10}K^4)$,
  \begin{equation*}
    T + 1 \geq \frac{- \log (2^{11} K^4 \eps ) }{\log( 2^{11/2} K^{5/2}) + \log(- \log (2^{11=} K^4 \eps ))}  \, . 
  \end{equation*}
\end{proof}

\subsubsection{Technical Lemmas}
\begin{lemma}\label{lem:invert-query}
  For any $a, b, x > 0$, for any $\eps \leq a / e$,
    \begin{equation*}
      \text{if} \quad \eps \geq a \, (1/ (bx))^x, 
      \quad \text{then} \quad 
      x \geq  \frac{\log( a/ \eps)}{ \log( b \log (a / \eps) )} \, . 
    \end{equation*}
\end{lemma}
\begin{proof}[\bfseries Proof of Lemma~\ref{lem:invert-query}]
    Assume
    \begin{equation*}
      \eps \geq a \, (1/ (bx))^x = e^{-x \log(bx)}
    \end{equation*}
    then 
    \begin{equation*}
      (\eps / a )^b \geq  e^{- bx \log(bx)} 
    \end{equation*}
    thus, applying logarithms to both sides,
    \begin{equation*}
      b \log (\eps / a) \geq - bx \log bx \, .
    \end{equation*}
    Reformulate as, since $a \geq e \eps $
    \begin{equation*}
      bx \log bx \geq b \log( a / \eps) \geq b > 0 \, .
    \end{equation*}
    We now apply Lambert's $W_0$ function, which is increasing on its main branch $[- 1 / e, + \infty)$, and use a standard lower bound on $W_0$ to obtain
    \begin{equation*}
      \log bx \geq W_0(-b \log (\eps / a))
       \geq \log( - b  \log (\eps / a)) - \log\log (- b \log( \eps / a))
    \end{equation*}
    thus 
    \begin{equation*}
      bx \geq - b \log( \eps  / a)  / \log( - b \log (\eps / a)) \, , 
    \end{equation*}
    which is the claimed bound. 
\end{proof}
\begin{lemma}\label{lem:ball}
  Let $\alpha \geq 0$ and $s > 0$ be real numbers, let $M \in \mathcal M_K(\R)$ be a matrix such that
  \begin{equation*}
      \max_{i, j\in[K]} \big| M_{i,j}- s I_K \big| \leq \alpha \, . 
  \end{equation*}
  Then for any $\eps$-Nash equilibrium $(p, q)$ of $M$ we have for any $i, j \in [K]$
  \begin{equation*}
        \min\!\big(p(i), \,  q(j)\big) \geq \frac{1}{K} - \frac{2(\alpha + \eps)(K-1)}{s} \, .
  \end{equation*}
\end{lemma}
%
%
\begin{proof}[\textbf{Proof}]
  Let $(p, q)$ denote an $\eps$-NE of $M$. Then 
  \begin{equation}\label{eq:pqequibrim}
      \max_{j\in[K]} (M^\top p)_j - \min_{i \in [K]} (M q)_i \leq 2\eps \, . 
  \end{equation}
  Now for any $j \in [K]$, since $M$ is close to $s I_{K}$, 
  \begin{equation*}
      (M^\top p)_j \geq s\,p(j) - \alpha
  \end{equation*}
  thus, using the fact that $p$ is a probability vector,
  \begin{equation}\label{eq:pmass}
      \max_{j \in [K]} (M^\top p)_j \geq s \max_{j \in [K]} p(j) - \alpha \geq \frac{s}{K} - \alpha \, . 
  \end{equation}
  Similarly, 
  \begin{equation}\label{eq:qmass}
      \min_{i\in [K]} (M q)_i \leq s \min_{i \in [K]} q(i) + \alpha
  \end{equation}
  Combining the equations \eqref{eq:pqequibrim}, \eqref{eq:pmass} and \eqref{eq:qmass} above, we have 
  \begin{equation*}
      \frac{s}{K} - \alpha - \eps \leq  s \min_{i \in [K]} q(i) + \alpha + \eps \, ,
  \end{equation*}
  i.e., after rearranging, 
  \begin{equation*}
      \min_{i \in [K]} q(i) \geq \frac{1}{K} - \frac{2(\alpha + \eps)}{s} \, . 
  \end{equation*}
  Similarly
  \begin{equation*}
      \max_{j \in [K]} p(j) \leq \frac{1}{s} \max_{j \in [K]} \big(M^\top p)_j + \frac{\alpha}{s} \leq \min_{i \in [K]}q(i) + 2 \frac{\alpha + \eps}{s} \leq \frac{1}{K} + \frac{2(\alpha + \eps)}{s}
  \end{equation*}
  Therefore
  \begin{multline*}
      \min_{j \in [K]} p(k) \geq 1 - (K-1) \max_{j \in [K]} p(j)
      \geq 1 - \frac{K-1}{K} - (K-1) \frac{2(\alpha + \eps)}{s} \\
      = \frac{1}{K} - (K-1) \frac{2(\alpha + \eps)}{s} \, .
  \end{multline*}
  Completing the proof.
\end{proof}

\begin{corollary}
  \label{cor:open-fully-mixed-eps-ne}
    For any $\eps \leq 1 / (16K^2)$, for any game matrix $M$ such that
    \begin{equation*}
        \max_{i, j \in [K] } \Big|M_{i, j} - \frac{1}{2}I_K\Big| \leq \frac{1}{16K^2}
    \end{equation*}
    all $\eps$-NE $(p, q)$ of $M$ satisfy for all $i, j \in [K]$
    \begin{equation*}
        \min \big(p(i), q(j)\big) \geq \frac{1}{2K} \, .
    \end{equation*}
\end{corollary}

\begin{corollary}
  \label{lem:open-fully-mixed-exact-ne}
    For any game matrix $M$ such that
    \begin{equation*}
        \max_{i, j \in [K]} \Big|M_{i, j} - \frac{1}{2}I_K\Big| \leq \frac{1}{16K^2}
    \end{equation*}
    all exact Nash equilibria $(p, q)$ are fully supported.
\end{corollary}

\end{document}

%% file: preamble.tex
\usepackage{standalone}
\usepackage{pgfplots}
\pgfplotsset{width=10cm,compat=1.8} 
\usepackage[english]{babel}
\usepackage{enumitem}

\usepackage{amsmath, amsthm, amssymb}
\usepackage{float}
\usepackage{pgf}
\usepackage{tikz}
\usetikzlibrary{shapes,external}
\usetikzlibrary{pgfplots.groupplots}
\tikzset{naming/.style={align=center,font=\footnotesize}}
\tikzset{area/.style = {draw, shape = regular polygon, regular polygon sides = 6, thick, minimum width = 5cm}}

\renewcommand{\epsilon}{\varepsilon}
\newcommand{\eps}{\varepsilon}

\renewcommand{\P}{\mathbb P}
\newcommand{\E}{\mathbb E}
\newcommand{\R}{\mathbb R}
\newcommand{\Q}{\mathbb Q}
\newcommand{\N}{\mathbb N}

\DeclareMathOperator*{\argmin}{argmin}

\DeclareMathOperator{\Span}{Span}
\DeclareMathOperator{\Proj}{Proj}

\DeclareMathOperator{\Tr}{Tr}

\DeclareMathOperator{\cX}{\mathcal{X}}
\DeclareMathOperator{\cY}{\mathcal{Y}}

\DeclareMathOperator{\one}{\textbf{1}}

\newtheorem*{theorem*}{Theorem}

\newtheorem{theorem}{Theorem}
\newtheorem{proposition}[theorem]{Proposition}
\newtheorem{remark}[theorem]{Remark}
\newtheorem{definition}[theorem]{Definition}
\newtheorem{corollary}[theorem]{Corollary}
\newtheorem{lemma}[theorem]{Lemma}
\newtheorem{example}[theorem]{Example}

\renewcommand{\leq}{\leqslant}
\renewcommand{\geq}{\geqslant}

\newcommand{\reals}{\mathbb{R}}
\newcommand{\ip}[1]{\langle #1 \rangle} 
\newcommand{\ones}{\mathbf 1}           
\newcommand{\half}{\frac{1}{2}}
\let\top\intercal 

%% file: neurips_2023_games_lower_bound.bbl
\begin{thebibliography}{43}
\providecommand{\natexlab}[1]{#1}
\providecommand{\url}[1]{\texttt{#1}}
\expandafter\ifx\csname urlstyle\endcsname\relax
  \providecommand{\doi}[1]{doi: #1}\else
  \providecommand{\doi}{doi: \begingroup \urlstyle{rm}\Url}\fi

\bibitem[Azizian et~al.(2020)Azizian, Mitliagkas, Lacoste-Julien, and
  Gidel]{azizian2020tight}
Wa{\"\i}ss Azizian, Ioannis Mitliagkas, Simon Lacoste-Julien, and Gauthier
  Gidel.
\newblock A tight and unified analysis of gradient-based methods for a whole
  spectrum of differentiable games.
\newblock In \emph{International conference on artificial intelligence and
  statistics}, pages 2863--2873. PMLR, 2020.

\bibitem[Babichenko(2016)]{babichenko2016query-complexit}
Yakov Babichenko.
\newblock Query complexity of approximate {N}ash equilibria.
\newblock \emph{Journal of the ACM (JACM)}, 63\penalty0 (4):\penalty0 1--24,
  2016.

\bibitem[Brown(1951)]{brown1951iterative-solut}
George~W Brown.
\newblock Iterative solution of games by fictitious play.
\newblock \emph{Act. Anal. Prod Allocation}, 13\penalty0 (1):\penalty0 374,
  1951.

\bibitem[Cai et~al.(2022)Cai, Oikonomou, and Zheng]{cai2022tight}
Yang Cai, Argyris Oikonomou, and Weiqiang Zheng.
\newblock Tight last-iterate convergence of the extragradient method for
  constrained monotone variational inequalities.
\newblock \emph{arXiv preprint arXiv:2204.09228}, 2022.

\bibitem[Cesa-Bianchi and Lugosi(2006)]{Cesa-Bianchi:2006uo}
Nicolo Cesa-Bianchi and G{\'a}bor Lugosi.
\newblock \emph{Prediction, learning, and games}.
\newblock Cambridge university press, 2006.

\bibitem[Chen and Peng(2020)]{Chen2020HedgingIG}
Xi~Chen and Binghui Peng.
\newblock Hedging in games: Faster convergence of external and swap regrets.
\newblock \emph{Advances in Neural Information Processing Systems},
  33:\penalty0 18990--18999, 2020.

\bibitem[Chen et~al.(2006)Chen, Deng, and Teng]{DBLP:journals/eccc/ChenDT06}
Xi~Chen, Xiaotie Deng, and Shang{-}Hua Teng.
\newblock Computing {N}ash equilibria: Approximation and smoothed complexity.
\newblock \emph{Electron. Colloquium Comput. Complex.}, 2006.

\bibitem[Chen et~al.(2009)Chen, Deng, and Teng]{10.1145/1516512.1516516}
Xi~Chen, Xiaotie Deng, and Shang-Hua Teng.
\newblock Settling the complexity of computing two-player {N}ash equilibria.
\newblock \emph{J. ACM}, 56\penalty0 (3), may 2009.

\bibitem[Daskalakis et~al.(2009{\natexlab{a}})Daskalakis, Goldberg, and
  Papadimitriou]{article}
Constantinos Daskalakis, Paul Goldberg, and Christos Papadimitriou.
\newblock The complexity of computing a {N}ash equilibrium.
\newblock \emph{SIAM J. Comput.}, 39:\penalty0 195--259, 02 2009{\natexlab{a}}.

\bibitem[Daskalakis et~al.(2009{\natexlab{b}})Daskalakis, Mehta, and
  Papadimitriou]{daskalakis2009note}
Constantinos Daskalakis, Aranyak Mehta, and Christos Papadimitriou.
\newblock A note on approximate {N}ash equilibria.
\newblock \emph{Theoretical Computer Science}, 410\penalty0 (17):\penalty0
  1581--1588, 2009{\natexlab{b}}.

\bibitem[Daskalakis et~al.(2011)Daskalakis, Deckelbaum, and
  Kim]{daskalakis2011near-optimal-no}
Constantinos Daskalakis, Alan Deckelbaum, and Anthony Kim.
\newblock Near-optimal no-regret algorithms for zero-sum games.
\newblock In \emph{Proceedings of the twenty-second annual ACM-SIAM symposium
  on Discrete Algorithms}, pages 235--254. SIAM, 2011.

\bibitem[Daskalakis et~al.(2015)Daskalakis, Deckelbaum, and
  Kim]{daskalakis2015near-optimal-no}
Constantinos Daskalakis, Alan Deckelbaum, and Anthony Kim.
\newblock Near-optimal no-regret algorithms for zero-sum games.
\newblock \emph{Games and Economic Behavior}, 92:\penalty0 327--348, 2015.

\bibitem[Daskalakis et~al.(2021{\natexlab{a}})Daskalakis, Fishelson, and
  Golowich]{daskalakis2021near-optimal-no}
Constantinos Daskalakis, Maxwell Fishelson, and Noah Golowich.
\newblock Near-optimal no-regret learning in general games.
\newblock In \emph{Advances in Neural Information Processing Systems},
  volume~34, pages 27604--27616, 2021{\natexlab{a}}.

\bibitem[Daskalakis et~al.(2021{\natexlab{b}})Daskalakis, Skoulakis, and
  Zampetakis]{daskalakis2021the-complexity-}
Constantinos Daskalakis, Stratis Skoulakis, and Manolis Zampetakis.
\newblock The complexity of constrained min-max optimization.
\newblock In \emph{Proceedings of the 53rd Annual ACM SIGACT Symposium on
  Theory of Computing}, pages 1466--1478, 2021{\natexlab{b}}.

\bibitem[Farina et~al.(2022)Farina, Anagnostides, Luo, Lee, Kroer, and
  Sandholm]{farina2022near-optimal-no}
Gabriele Farina, Ioannis Anagnostides, Haipeng Luo, Chung-Wei Lee, Christian
  Kroer, and Tuomas Sandholm.
\newblock Near-optimal no-regret learning dynamics for general convex games.
\newblock \emph{Advances in Neural Information Processing Systems},
  35:\penalty0 39076--39089, 2022.

\bibitem[Fearnley et~al.(2015)Fearnley, Gairing, Goldberg, and
  Savani]{fearnley2015learning-equili}
John Fearnley, Martin Gairing, Paul~W Goldberg, and Rahul Savani.
\newblock Learning equilibria of games via payoff queries.
\newblock \emph{Journal of Machine Learning Research}, 16:\penalty0 1305--1344,
  2015.

\bibitem[Freund and Schapire(1999)]{freund1999adaptive}
Yoav Freund and Robert~E Schapire.
\newblock Adaptive game playing using multiplicative weights.
\newblock \emph{Games and Economic Behavior}, 29\penalty0 (1-2):\penalty0
  79--103, 1999.

\bibitem[Gidel et~al.(2017)Gidel, Jebara, and Lacoste-Julien]{gidel2017frank}
Gauthier Gidel, Tony Jebara, and Simon Lacoste-Julien.
\newblock Frank-wolfe algorithms for saddle point problems.
\newblock In \emph{Artificial Intelligence and Statistics}, pages 362--371.
  PMLR, 2017.

\bibitem[Gilpin et~al.(2008)Gilpin, Pe\~{n}a, and
  Sandholm]{10.5555/1619995.1620009}
Andrew Gilpin, Javier Pe\~{n}a, and Thomas Sandholm.
\newblock First-order algorithm with $o(\ln(1/\epsilon))$ convergence for
  $\epsilon$-equilibrium in two-person zero-sum games.
\newblock In \emph{Proceedings of the 23rd National Conference on Artificial
  Intelligence}, AAAI'08, pages 75--82, 2008.

\bibitem[Golowich et~al.(2020)Golowich, Pattathil, Daskalakis, and
  Ozdaglar]{golowich2020iterate}
Noah Golowich, Sarath Pattathil, Constantinos Daskalakis, and Asuman Ozdaglar.
\newblock Last iterate is slower than averaged iterate in smooth convex-concave
  saddle point problems, 2020.

\bibitem[Hart and Nisan(2018)]{hart2018query}
Sergiu Hart and Noam Nisan.
\newblock The query complexity of correlated equilibria.
\newblock \emph{Games and Economic Behavior}, 108:\penalty0 401--410, 2018.

\bibitem[Hazan and Koren(2016)]{hazan2016the-computation}
Elad Hazan and Tomer Koren.
\newblock The computational power of optimization in online learning.
\newblock In \emph{Proceedings of the Forty-Eighth Annual ACM Symposium on
  Theory of Computing}, STOC '16, pages 128--141, 2016.

\bibitem[Hsieh et~al.(2021)Hsieh, Antonakopoulos, and
  Mertikopoulos]{Hsieh2021AdaptiveLI}
Yu-Guan Hsieh, Kimon Antonakopoulos, and P.~Mertikopoulos.
\newblock Adaptive learning in continuous games: Optimal regret bounds and
  convergence to {N}ash equilibrium.
\newblock In \emph{Annual Conference Computational Learning Theory}, 2021.

\bibitem[Ibrahim et~al.(2020)Ibrahim, Azizian, Gidel, and
  Mitliagkas]{IbrahimAGM20}
Adam Ibrahim, Wa{\"\i}ss Azizian, Gauthier Gidel, and Ioannis Mitliagkas.
\newblock Linear lower bounds and conditioning of differentiable games.
\newblock In \emph{Proceedings of the 37th International Conference on Machine
  Learning, ICML 2020}, volume 119, pages 4583--4593. PMLR, 2020.

\bibitem[Jiang and Leyton-Brown(2011)]{jiang2011polynomial-time}
Albert~Xin Jiang and Kevin Leyton-Brown.
\newblock Polynomial-time computation of exact correlated equilibrium in
  compact games.
\newblock In \emph{Proceedings of the 12th ACM conference on Electronic
  commerce}, pages 119--126, 2011.

\bibitem[Korpelevich(1976)]{korpelevich1976extragradient}
Galina~M Korpelevich.
\newblock The extragradient method for finding saddle points and other
  problems.
\newblock \emph{Matecon}, 12:\penalty0 747--756, 1976.

\bibitem[Lin et~al.(2020)Lin, Jin, and Jordan]{Lin2020NearOptimalAF}
Tianyi Lin, Chi Jin, and Michael~I. Jordan.
\newblock Near-optimal algorithms for minimax optimization.
\newblock In \emph{Annual Conference Computational Learning Theory}, 2020.

\bibitem[Mokhtari et~al.(2020)Mokhtari, Ozdaglar, and
  Pattathil]{mokhtari2020convergence-rat}
Aryan Mokhtari, Asuman~E Ozdaglar, and Sarath Pattathil.
\newblock Convergence rate of o(1/k) for optimistic gradient and extragradient
  methods in smooth convex-concave saddle point problems.
\newblock \emph{SIAM Journal on Optimization}, 30\penalty0 (4):\penalty0
  3230--3251, 2020.

\bibitem[Nemirovsky(1991)]{nemirovsky1991on-optimality-o}
Arkadi~S Nemirovsky.
\newblock On optimality of {K}rylov's information when solving linear operator
  equations.
\newblock \emph{Journal of Complexity}, 7\penalty0 (2):\penalty0 121--130,
  1991.

\bibitem[Nemirovsky(1992)]{nemirovsky1992information-bas}
Arkadi~S Nemirovsky.
\newblock Information-based complexity of linear operator equations.
\newblock \emph{Journal of Complexity}, 8\penalty0 (2):\penalty0 153--175,
  1992.

\bibitem[Nesterov(2014)]{Nesterov2014IntroductoryLO}
Yurii Nesterov.
\newblock Introductory lectures on convex optimization - a basic course.
\newblock In \emph{Applied Optimization}, 2014.

\bibitem[Orabona and P{\'a}l(2018)]{ORABONA201850}
Francesco Orabona and D{\'a}vid P{\'a}l.
\newblock Scale-free online learning.
\newblock \emph{Theoretical Computer Science}, 716:\penalty0 50--69, 2018.
\newblock Special Issue on ALT 2015.

\bibitem[Ouyang and Xu(2021)]{Ouyang:2021aa}
Yuyuan Ouyang and Yangyang Xu.
\newblock Lower complexity bounds of first-order methods for convex-concave
  bilinear saddle-point problems.
\newblock \emph{Mathematical Programming}, 185\penalty0 (1):\penalty0 1--35,
  2021.

\bibitem[Papadimitriou and Roughgarden(2008)]{papadimitriou2008computing-corre}
Christos~H Papadimitriou and Tim Roughgarden.
\newblock Computing correlated equilibria in multi-player games.
\newblock \emph{Journal of the ACM (JACM)}, 55\penalty0 (3):\penalty0 1--29,
  2008.

\bibitem[Piliouras et~al.(2022)Piliouras, Sim, and
  Skoulakis]{piliouras2022beyond}
Georgios Piliouras, Ryann Sim, and Stratis Skoulakis.
\newblock Beyond time-average convergence: Near-optimal uncoupled online
  learning via clairvoyant multiplicative weights update.
\newblock \emph{Advances in Neural Information Processing Systems},
  35:\penalty0 22258--22269, 2022.

\bibitem[Popov(1980)]{popov1980a-modification-}
Leonid~Denisovich Popov.
\newblock A modification of the {A}rrow-{H}urwicz method for search of saddle
  points.
\newblock \emph{Mathematical notes of the Academy of Sciences of the USSR},
  28:\penalty0 845--848, 1980.

\bibitem[Rakhlin and Sridharan(2013)]{rakhlin2013optimization}
Sasha Rakhlin and Karthik Sridharan.
\newblock Optimization, learning, and games with predictable sequences.
\newblock \emph{Advances in Neural Information Processing Systems}, 26, 2013.

\bibitem[Robinson(1951)]{robinson1951an-iterative-me}
Julia Robinson.
\newblock An iterative method of solving a game.
\newblock \emph{Annals of mathematics}, pages 296--301, 1951.

\bibitem[Stoltz and Lugosi(2007)]{stoltz2007learning-correl}
Gilles Stoltz and G{\'a}bor Lugosi.
\newblock Learning correlated equilibria in games with compact sets of
  strategies.
\newblock \emph{Games and Economic Behavior}, 59\penalty0 (1):\penalty0
  187--208, 2007.

\bibitem[Syrgkanis et~al.(2015)Syrgkanis, Agarwal, Luo, and
  Schapire]{syrgkanis2015fast}
Vasilis Syrgkanis, Alekh Agarwal, Haipeng Luo, and Robert~E Schapire.
\newblock Fast convergence of regularized learning in games.
\newblock \emph{Advances in Neural Information Processing Systems}, 28, 2015.

\bibitem[Wei et~al.(2020)Wei, Lee, Zhang, and Luo]{Wei2020LinearLC}
Chen-Yu Wei, Chung-Wei Lee, Mengxiao Zhang, and Haipeng Luo.
\newblock Linear last-iterate convergence in constrained saddle-point
  optimization.
\newblock In \emph{International Conference on Learning Representations}, 2020.

\bibitem[Yang et~al.(2022)Yang, Jordan, and Chavdarova]{yang2022solving}
Tong Yang, Michael~I Jordan, and Tatjana Chavdarova.
\newblock Solving constrained variational inequalities via an interior point
  method.
\newblock \emph{arXiv preprint arXiv:2206.10575}, 2022.

\bibitem[Zhang et~al.(2022)Zhang, Hong, and Zhang]{zhang2022on-lower-iterat}
Junyu Zhang, Mingyi Hong, and Shuzhong Zhang.
\newblock On lower iteration complexity bounds for the convex concave saddle
  point problems.
\newblock \emph{Mathematical Programming}, 194\penalty0 (1-2):\penalty0
  901--935, 2022.

\end{thebibliography}
